%% file: main.tex
\documentclass[final,5p,times,twocolumn,nopreprintline]{elsarticle}
\usepackage{amsmath,amssymb,amsthm,mathtools,mathrsfs,booktabs}
\usepackage{microtype}
\usepackage{array} 
\usepackage{xurl}
\usepackage[hidelinks]{hyperref}
\usepackage{xcolor}
\usepackage{CJKutf8}
\makeatletter
\let\ps@pprintTitle\ps@plain
\makeatother
\newtheorem{theorem}{Theorem}
\newtheorem{lemma}[theorem]{Lemma}

\newtheorem{corollary}[theorem]{Corollary}
\theoremstyle{definition}
\newtheorem{definition}[theorem]{Definition}
\theoremstyle{remark}
\newtheorem{remark}[theorem]{Remark}
\newcommand{\R}{\mathbb R}
\newcommand{\V}{\mathcal V}
\newcommand{\Vwa}{\mathcal V^{\mathrm{wa}}}
\newcommand{\Vwm}{\mathcal V^{\mathrm{wm}}}
\newcommand{\one}{\mathbf 1}
\newcommand{\zero}{\mathbf 0}
\newcommand{\wa}{\mathrm{wa}}

\DeclareMathOperator{\Fix}{Fix}
\DeclareMathOperator{\Med}{Med}
\DeclareMathOperator*{\argmin}{arg\,min}
\DeclareMathOperator*{\argmax}{arg\,max}

\begin{document}
\begin{frontmatter}
\title{Equilibria and Consensus in Mixed Weighted-Average and Weighted-Median Opinion Dynamics: The Consensus-Promoting Role of Weighted-Average Individuals}
\author[pku]{Junkai Cheng}
\ead{junkai20031004@gmail.com}
\author[pku]{Wenjun Mei\corref{cor1}}
\ead{mei@pku.edu.cn}
\cortext[cor1]{Corresponding author.}
\address[pku]{School of Advanced Manufacturing and Robotics, Peking University,\\
ENN Engineering Building, No.~116 Zhongguancun North Avenue, Haidian District, Beijing 100871, China}
\begin{abstract}
How individuals respond to their neighbors' opinions shapes whether a networked group reaches consensus. This paper analyzes synchronous discrete-time opinion dynamics on a fixed influence network, where individuals use either weighted-average (WA) or weighted-median (WM) opinion-aggregation mechanisms with individual-specific inertia coefficients. We characterize equilibria through graph-theoretic conditions on opinion level sets and WA balance equations, proving that every solution converges to a Lyapunov-stable equilibrium. We define type-dependent cohesive sets, whose WA members assign all their influence weight to the set and WM members assign at least half. The absence of two disjoint such sets gives a necessary and sufficient graph-theoretic criterion for consensus from any initial condition. Otherwise, an open set of initial conditions leads to persistent dissensus. Converting WM individuals to WA cannot create type-dependent cohesive sets and preserves existing guarantees of consensus from any initial condition. If the conversion eliminates every disjoint pair of these sets, the converted system also guarantees consensus from any initial condition. Four numerical studies examine the locations and proportion of weighted-average individuals, as well as the influence weights, while comparing the consensus criterion with finite-horizon simulation outcomes. In the four network ensembles with 100 individuals, the observed frequency of reaching a prescribed consensus tolerance is nondecreasing with WA proportion at each tested simulation horizon. This trend supports the consensus-promoting role of WA individuals, although individual realizations need not follow it.
\end{abstract}
\end{frontmatter}

\section{Introduction}\label{sec:introduction}
\subsection{Background and motivation}
Opinion dynamics studies how individuals change their opinions through interpersonal influence and how network structure shapes the resulting collective behavior. Mathematical models connect individual opinion update rules with outcomes such as consensus and persistent dissensus~\cite{degroot1974,mei2022}. An important modeling question is how an individual combines the opinions received from others: changing the opinion-aggregation mechanism can change whether a group eventually reaches consensus, even when the influence network remains fixed~\cite{mei2022}.

Two relevant opinion-aggregation mechanisms are weighted-average (WA)~\cite{french1956,degroot1974} and weighted-median (WM)~\cite{mei2022}. The former represents compromise through an influence-weighted mean, whereas the latter can sustain distinct opinion groups on a fixed network~\cite{mei2022,mei2024}. A mixed population, in which some individuals use the WA mechanism and others use the WM mechanism, has also been studied~\cite{wu2025}. However, the consensus theorem in~\cite{wu2025} assumes no influence from WM to WA individuals and requires additional conditions on the WA block and WM cohesive sets. This motivates us to study mixed WA--WM dynamics on a general fixed influence network without imposing the structural restrictions. Given that WM dynamics can sustain separated opinion groups, can the presence of WA individuals prevent such separation, and how does the outcome depend on their locations and proportion?

Individuals may also differ in how much of their current opinions they retain at each update, as in inertial bounded-confidence models~\cite{chazelle2017,li2022}. This raises the question of how such individual-specific inertia interacts with heterogeneous opinion-aggregation mechanisms. To understand these dynamics, it is important to characterize the possible equilibria and determine whether solutions converge. More interestingly, when different mechanisms coexist and individuals have different inertia coefficients, can certain individuals play a consensus-promoting role in preventing persistent dissensus? We discuss these questions under a mixed WA--WM model with synchronous discrete-time updates, individual-specific inertia coefficients, and a general fixed row-stochastic adjacency matrix.

\subsection{Literature review}

\paragraph{Weighted-average models}
The French--DeGroot model~\cite{french1956,degroot1974} is a classical weighted-average opinion model in which individuals repeatedly update their opinions through weighted averaging over a fixed influence network. Consensus is governed by properties of the adjacency matrix. Subsequent models extend this framework by modifying the interaction rules and influence structures. In bounded-confidence models, individuals interact only with others whose opinions are sufficiently close: Hegselmann--Krause dynamics~\cite{hegs2002} average opinions within confidence neighborhoods, whereas Deffuant--Weisbuch dynamics~\cite{deffuant2000} use pairwise compromise. Altafini's model~\cite{altafini2013} further considers signed influence weights to represent antagonistic interactions.

\paragraph{Weighted-median models}
The weighted-average mechanism has an unrealistic implication: the attractiveness of an opinion to an individual is proportional to the distance between their opinions. This leads to overly strong attractions between distant opinions, driving the system toward consensus under mild network connectivity conditions and potentially overriding the effects of network structure on local opinion clustering~\cite{mei2022}. To address this limitation, Mei et al.~\cite{mei2022} introduce the weighted-median (WM) mechanism as an alternative micro-foundation of opinion dynamics, derived from cognitive dissonance theory. By minimizing a weighted sum of absolute opinion differences, WM responses depend on the ordering and influence weights of received opinions rather than the numerical distances between them. This mechanism does not necessarily lead to consensus and can capture local opinion clusters. Empirical and numerical studies further demonstrate its advantages in predicting individual opinion shifts and reproducing real-world patterns of opinion evolution~\cite{mei2022}. Subsequent work develops a graph-theoretic analysis based on cohesive sets to characterize equilibria and study convergence, consensus, and dissensus under random asynchronous updating~\cite{mei2024}. Further extensions consider fixed prejudices under synchronous updates~\cite{zhang2025}, continuous-time dynamics~\cite{han2025}, and higher-order interactions~\cite{chen2026}.

\paragraph{Heterogeneity in opinion dynamics}
Heterogeneity in opinion dynamics takes several forms. Individuals may differ in their susceptibility to social influence and attachment to initial opinions, as in the Friedkin--Johnsen model~\cite{friedkin1990}. Other forms of heterogeneity concern individual confidence and influence ranges, as studied by Mirtabatabaei and Bullo~\cite{mirtabatabaei2012}. Beyond these differences in parameter values, individuals may use different opinion-aggregation mechanisms, such as the WA and WM mechanisms. They may also retain different fractions of their current opinions at each update. Chazelle and Wang~\cite{chazelle2017} analyze inertial Hegselmann--Krause systems, while Li~\cite{li2022} studies individual, time-varying weights on current opinions in mixed Hegselmann--Krause dynamics. Unlike attachment to a fixed initial opinion, this form of inertia retains part of an opinion that evolves over time.

\paragraph{Mixed WA--WM models}
Since both WA and WM responses minimize weighted opinion-difference costs, using squared and absolute deviations, respectively, it is natural to consider their coexistence on the same network. Wu et al.~\cite{wu2025} study synchronous mixed WA--WM dynamics, alongside equilibrium and periodic behavior in pure WM dynamics. In the mixed model, fixed groups of individuals use weighted-average or weighted-median opinion-aggregation mechanisms, without a separate individual-specific inertia term. They establish asymptotic consensus from any initial opinion vector, but their theorem requires a block-triangular adjacency structure in which WA individuals receive no influence from WM individuals. It additionally requires the WA adjacency block to be irreducible and aperiodic and excludes nonempty cohesive sets consisting solely of WM individuals. Under these conditions, the common limiting opinion is a weighted average of the WA individuals' initial opinions, with weights determined by the WA adjacency block. Thus, their consensus result depends on specific structural assumptions and does not cover general fixed influence networks.

\subsection{Statement of contributions}
The main contributions are summarized as follows.

First, we formulate synchronous discrete-time WA--WM dynamics with individual-specific inertia on a fixed influence network, allowing individuals to use different opinion-aggregation mechanisms.

Second, we characterize equilibria by combining graph-theoretic conditions on opinion level sets with WA balance equations and prove that every solution converges to a Lyapunov-stable equilibrium. A graph-theoretic criterion based on type-dependent cohesive sets characterizes consensus from any initial condition; when two such sets are disjoint, an open set of initial conditions leads to persistent dissensus. We further show how converting WM individuals to WA can preserve or establish guarantees of consensus from any initial condition, explaining the consensus-promoting role of WA individuals.

Third, four numerical studies examine how WA locations, WA proportions, and influence weights affect consensus, and compare theoretical consensus guarantees with finite-horizon simulations. In the four network ensembles with 100 individuals, the observed frequency of reaching a prescribed consensus tolerance is nondecreasing with WA proportion at each tested simulation horizon, although individual realizations need not follow this trend.

\paragraph{Organization}
The remainder of the paper is organized as follows. Section~\ref{sec:model} introduces the model and problems of interest. Section~\ref{sec:theory} presents the theoretical analysis. Section~\ref{sec:numerics} reports the numerical studies. Section~\ref{sec:conclusion} concludes the paper. A separate supplement provides additional details and results for the numerical studies.

\begin{figure*}[t]
 \centering
 \includegraphics[width=\textwidth]{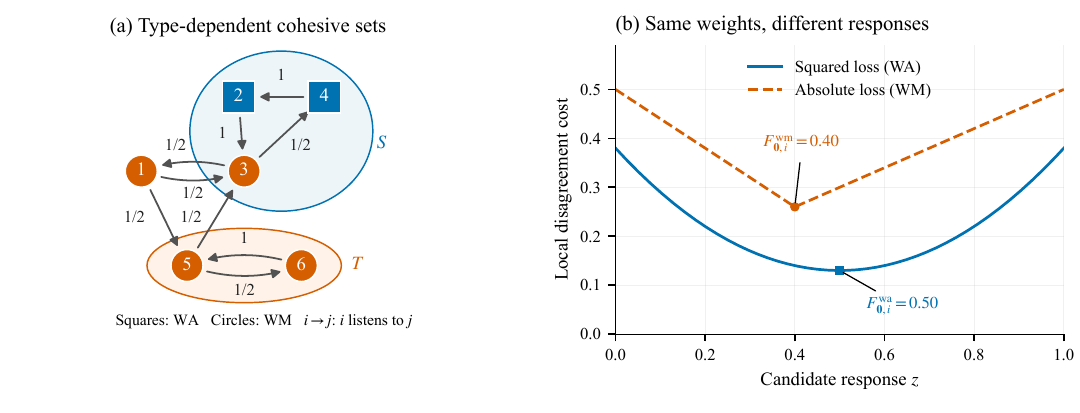}
 \caption{Weighted-average and weighted-median opinion-aggregation mechanisms. (a) Six-node influence network: blue squares $\{2,4\}$ represent WA individuals and orange circles represent WM individuals. Each arrow $i\to j$ means that $i$ listens to $j$, with influence weight $a_{ij}$. The shaded sets $S=\{2,3,4\}$ and $T=\{5,6\}$ are type-dependent cohesive sets (Definition~\ref{def:cohesion}): each WA member assigns all its weight to its set, and each WM member assigns at least one half. (b) Local costs for opinions $(0,0.4,1)$ with weights $(0.2,0.5,0.3)$: the solid quadratic curve is minimized at $0.5$, and the dashed absolute-error curve at $0.4$. Markers indicate the corresponding minimizers.}\label{fig:model}
\end{figure*}

\section{Notation and model setup}\label{sec:model}
\subsection{Network and local opinion aggregation}
Consider $n\geq 2$ individuals indexed by $\V=\{1,\ldots,n\}$, interacting over a fixed influence network, that is, a directed and weighted graph whose nodes are indexed by $\V$ and represent the corresponding individuals. The network is represented by a row-stochastic adjacency matrix $A=(a_{ij})\in\R^{n\times n}$:
\begin{equation}\label{eq:stochastic}
 a_{ij}\geq 0,\qquad \sum_{j\in\V}a_{ij}=1\quad(i\in\V).
\end{equation}
The influence weight $a_{ij}$ represents the influence of individual $j$ on individual $i$. We draw a directed edge $(i,j)$ when $a_{ij}>0$, so an edge from $i$ to $j$ means that $i$ listens to $j$. Self-weights are allowed. For any $S\subseteq\V$, write
\begin{equation}\label{eq:mass}
 a_i(S)=\sum_{j\in S}a_{ij}.
\end{equation}
Thus $a_i(S)$ is the total influence weight that individual $i$ assigns to the members of $S$. Write $\mathbb N_0=\{0,1,\ldots\}$ and $\mathbb N_+=\{1,2,\ldots\}$. For any $i\in\mathcal V$ and $k\in\mathbb N_0$, the opinion held by individual $i$ at time step $k$ is denoted by $x_i(k)$. Denote by $x(k)=(x_1(k),\ldots,x_n(k))^\top\in\mathbb R^n$ the system's opinion vector at time step $k$. Let $\mathbf{1}$ denote the vector of ones. Inequalities between vectors are understood componentwise. For any map $G:\mathbb R^n\to\mathbb R^n$, write $\operatorname{Fix}(G)=\{x\in\mathbb R^n\mid G(x)=x\}$ as the set of $G$'s fixed points.

Denote by $\mathcal V^{\mathrm{wa}}$ and $\mathcal V^{\mathrm{wm}}$ the index sets of individuals that update their opinions according to the weighted-average (WA) mechanism~\cite{french1956,degroot1974} and the weighted-median (WM) mechanism~\cite{mei2022}, respectively, referred to as WA individuals and WM individuals, respectively. The two sets form a fixed disjoint partition of $\V$, i.e., $\V=\Vwa\cup\Vwm$ and $\Vwa\cap\Vwm=\varnothing$, and either subset may be empty. We call the partition $(\Vwa,\Vwm)$ a WA--WM configuration.

For a WA individual $i\in\mathcal V^{\mathrm{wa}}$, define the weighted-average response as the unique minimizer of the weighted squared-error cost:
\begin{equation}\label{eq:waresponse}
\underset{z\in\mathbb R}{\arg\min}\sum_{j\in\mathcal V}a_{ij}(z-x_j)^2=\sum_{j\in\mathcal V}a_{ij}x_j.
\end{equation}

For a WM individual $i\in\mathcal V^{\mathrm{wm}}$, define the weighted-median interval
\begin{equation}\label{eq:medianinterval}
\mathcal M_i(x;A)=\underset{z\in\mathbb R}{\arg\min}\sum_{j\in\mathcal V}a_{ij}|z-x_j|=[\ell_i(x;A),u_i(x;A)].
\end{equation}
Thus the weighted-median interval is the entire minimizer set of the weighted absolute-error cost; $\ell_i(x;A)$ and $u_i(x;A)$ are its lower and upper endpoints, respectively. Both endpoints are finite because the weights $a_{i1},\ldots,a_{in}$ are nonnegative and sum to one. The weighted-median response is the Euclidean projection~\cite[Section~8.1.1]{boyd2004} of individual $i$'s current opinion onto this interval:
\begin{equation}\label{eq:medianselection}
\operatorname{Med}_i(x;A)=\max\{\ell_i(x;A),\min\{x_i,u_i(x;A)\}\}.
\end{equation}
That is, if $x_i\in\mathcal M_i(x;A)$, then $\operatorname{Med}_i(x;A)=x_i$; otherwise, $\operatorname{Med}_i(x;A)$ is the endpoint of the interval $\mathcal M_i(x;A)$ that is closer to $x_i$. Equivalently,
\[
\operatorname{Med}_i(x;A)
=
\operatorname*{arg\,min}_{z\in\mathcal M_i(x;A)}
|x_i-z|.
\]

Combining the weighted-average response in~\eqref{eq:waresponse} with the weighted-median response in~\eqref{eq:medianselection}, we define the base response map $F_\zero:\R^n\to\R^n$:
\begin{equation}\label{eq:response}
 F_{\zero,i}(x)=
 \begin{cases}
 \displaystyle\sum_{j\in\V}a_{ij}x_j,&\text{if }i\in\Vwa,\\[2pt]
 \Med_i(x;A),&\text{if }i\in\Vwm.
 \end{cases}
\end{equation}
Here $\zero\in\R^n$ is the zero vector; the subscript denotes the zero-inertia case introduced below. The component $F_{\zero,i}(x)$ is the response of individual $i$ before inertia is applied. Figure~\ref{fig:model} illustrates the two local costs and the coexistence of WA and WM individuals on one network.

\subsection{Discrete-time dynamics and problem statement}
For fixed inertia coefficients $\gamma=(\gamma_1,\ldots,\gamma_n)^\top\in(0,1)^n$, consider the synchronous discrete-time system
\begin{equation}\label{eq:system}
 x_i(k+1)=F_{\gamma,i}(x(k))
 =\gamma_i x_i(k)+(1-\gamma_i)F_{\zero,i}(x(k)),\quad i\in\V.
\end{equation}
All individuals compute their updated opinions from the same opinion vector $x(k)$. The inertia coefficient $\gamma_i$ is the weight retained on the current opinion, while $1-\gamma_i$ is the weight assigned to the response. This retained opinion changes with time; it is not a fixed prejudice. The opinion update map $F_\gamma:\R^n\to\R^n$ gives the vector form $x(k+1)=F_\gamma(x(k))$ of the dynamics, starting from an arbitrary $x(0)\in\R^n$. The influence network, WA--WM configuration, and inertia coefficients remain fixed along each solution. When all individuals use WA, system~\eqref{eq:system} is a French--DeGroot model with adjacency matrix $\Gamma+(I-\Gamma)A$, where $\Gamma=\operatorname{diag}(\gamma_1,\ldots,\gamma_n)$.

For further interpretation, we extend the admissible range of the inertia coefficients to $\gamma\in[0,1]^n$, thereby including the boundary values. Setting $\gamma=\zero$ removes the added inertia and recovers exactly the base response map $F_\zero$ in~\eqref{eq:response}. At the other endpoint, $\gamma_i=1$ freezes individual $i$ at its current opinion. These boundary cases explain the notation but are outside the assumptions of the main results.

An opinion vector is a consensus vector (equivalently, a constant opinion vector) if all its components are equal, that is, if it has the form $c\one$ for some $c\in\R$; an opinion vector with unequal components is nonconstant. An opinion vector $x^\star\in\R^n$ is an equilibrium if and only if $F_\gamma(x^\star)=x^\star$. An equilibrium is a consensus equilibrium if its components are all equal and a dissensus equilibrium otherwise. The system reaches consensus asymptotically if $x(k)\to c\one$ for some $c\in\R$; convergence to a dissensus equilibrium represents persistent dissensus. Define the minimum opinion $m(x)$, maximum opinion $M(x)$, and opinion diameter $D(x)$ by
\begin{align}
 m(x)&=\min_{i\in\V}x_i,& M(x)&=\max_{i\in\V}x_i,\nonumber\\
 D(x)&=M(x)-m(x).&&\label{eq:diameter}
\end{align}

\paragraph{Problem statement}
For system~\eqref{eq:system}, we investigate two main questions: (i) how to characterize its equilibria and whether every solution is well-defined, bounded, and converges to a Lyapunov-stable equilibrium; and (ii) what graph-theoretic conditions guarantee consensus from any initial condition or allow persistent dissensus for some initial conditions.

\section{Theoretical analysis}\label{sec:theory}
In this section, we establish basic properties of $F_\zero$ and $F_\gamma$, characterize the equilibria of system~\eqref{eq:system}, prove convergence of every solution, and derive necessary and sufficient graph-theoretic conditions for consensus from any initial condition.
\subsection{Monotonicity, additive homogeneity, and boundedness}
\begin{lemma}[Weighted-median characterization~\cite{mei2024}]\label{lem:mediancharacterization}
For any $i\in\Vwm$, $x\in\R^n$, and $z\in\R$,
the scalar $z$ belongs to $\mathcal M_i(x;A)$ if and only if
\begin{equation}\label{eq:mediancondition}
\begin{cases}
\displaystyle\sum_{j:x_j<z}a_{ij}\leq\tfrac12,\\[2pt]
\displaystyle\sum_{j:x_j>z}a_{ij}\leq\tfrac12.
\end{cases}
\end{equation}
\end{lemma}

\begin{proof}
Let $\phi_i(z)=\sum_{j\in\V}a_{ij}|z-x_j|$.
Its left and right derivatives are
\[
\begin{aligned}
\phi'_{i,-}(z)
&=2\sum_{j:x_j<z}a_{ij}-1,\\
\phi'_{i,+}(z)
&=1-2\sum_{j:x_j>z}a_{ij}.
\end{aligned}
\]
By convexity, $z$ minimizes $\phi_i$ if and only if
$\phi'_{i,-}(z)\leq0\leq\phi'_{i,+}(z)$,
which is equivalent to~\eqref{eq:mediancondition}.
\end{proof}

Lemma~\ref{lem:mediancharacterization} allows us to establish the following properties of $F_\zero$ and $F_\gamma$.

\begin{lemma}[Well-definedness, monotonicity, additive homogeneity, and boundedness]\label{lem:order}
Consider system~\eqref{eq:system} with a fixed nonnegative row-stochastic adjacency matrix $A$ and a fixed WA--WM configuration $(\Vwa,\Vwm)$. For any $\gamma\in(0,1)^n$, the base response map $F_\zero$ in~\eqref{eq:response} and the opinion update map $F_\gamma$ in~\eqref{eq:system} satisfy the following properties.
\begin{enumerate}
\item Both maps are well-defined and single-valued on $\R^n$. Consequently, for any initial condition $x(0)\in\R^n$, system~\eqref{eq:system} admits a unique solution defined for all $k\in\mathbb N_0$.
\item Both maps are monotone and additively homogeneous: if $x\leq y$, then
$F_\zero(x)\leq F_\zero(y)$ and
$F_\gamma(x)\leq F_\gamma(y)$. Moreover,
\begin{equation}\label{eq:ordertranslation}
\begin{aligned}
F_\zero(x+c\one)&=F_\zero(x)+c\one,\\
F_\gamma(x+c\one)&=F_\gamma(x)+c\one
\end{aligned}
\qquad(x\in\R^n,\ c\in\R).
\end{equation}
In addition, both maps fix every consensus vector: $F_\zero(c\one)=F_\gamma(c\one)=c\one$.
\item For any $x\in\R^n$,
\begin{equation}\label{eq:hull}
 m(x)\leq m(F_\gamma(x))\leq M(F_\gamma(x))\leq M(x).
\end{equation}
Consequently, $m(x(k))$ is nondecreasing and $M(x(k))$ is nonincreasing along every solution. The initial box $[m(x(0)),M(x(0))]^n$ is positively invariant: $F_\gamma$ maps this box into itself. Thus all components remain in the initial opinion interval, and every solution is bounded.
\end{enumerate}
\end{lemma}
\begin{proof}
We first verify that both maps are well defined. For each WA individual,~\eqref{eq:waresponse} gives one finite response. For each WM individual, the absolute-error cost in~\eqref{eq:medianinterval} is continuous, convex, and tends to infinity as $|z|\to\infty$, since its weights sum to one. Its minimizers therefore form a nonempty compact interval. Projection onto that interval gives exactly one value in~\eqref{eq:medianselection}. Thus $F_\zero(x)$ and then $F_\gamma(x)$ are uniquely defined for any $x\in\R^n$. Applying $F_\gamma$ successively establishes existence and uniqueness of the solution at every time step.

We next prove monotonicity and additive homogeneity.
For any $i\in\Vwa$, since $a_{ij}\geq0$ for any $j\in\V$, we have that, whenever $x\leq y$,
\[
 \sum_j a_{ij}x_j\leq\sum_j a_{ij}y_j\quad\text{and}\quad
 \sum_j a_{ij}(x_j+c)=\sum_j a_{ij}x_j+c.
\]
For $i\in\Vwm$ at opinion vector $x$, define the cumulative weight
\begin{equation}\label{eq:cumulativeweight}
 C_i(t;x)=\sum_{j:x_j\leq t}a_{ij},\qquad t\in\R.
\end{equation}
Here $t$ is an opinion threshold. The value $C_i(t;x)$ is the total weight that individual $i$ assigns to opinions at or below $t$ in the specified opinion vector $x$. By~\eqref{eq:mediancondition}, the lower and upper endpoints of this individual's weighted-median interval $\mathcal M_i(x;A)$ are
\begin{equation}\label{eq:medianendpoints}
 \begin{aligned}
 \ell_i(x;A)&=\inf\{t:C_i(t;x)\geq1/2\},\\
 u_i(x;A)&=\inf\{t:C_i(t;x)>1/2\}.
 \end{aligned}
\end{equation}
To interpret the strict inequality in the second equation above, list the distinct opinions carrying positive weight in row $i$ as $v_1<\cdots<v_p$, combining weights at repeated opinions. If the cumulative weight $C_i(t;x)$ first reaches exactly $1/2$ at $v_r$, then $r<p$, $\ell_i(x;A)=v_r$, and $u_i(x;A)=v_{r+1}$. Every point in the interval $[v_r,v_{r+1}]$ minimizes the weighted absolute-error cost in~\eqref{eq:medianinterval}. If the cumulative weight instead crosses $1/2$ strictly at $v_r$, then $\ell_i(x;A)=u_i(x;A)=v_r$. The formulas therefore include repeated opinions, zero weights, and exact half-weight ties.

For $x\leq y$, the inclusion $\{j:y_j\leq t\}\subseteq\{j:x_j\leq t\}$ gives $C_i(t;x)\geq C_i(t;y)$ for any $t$. Consequently,
\[
 \begin{aligned}
 \{t:C_i(t;y)\geq1/2\}&\subseteq\{t:C_i(t;x)\geq1/2\},\\
 \{t:C_i(t;y)>1/2\}&\subseteq\{t:C_i(t;x)>1/2\}.
 \end{aligned}
\]
Taking infima yields $\ell_i(x;A)\leq\ell_i(y;A)$ and $u_i(x;A)\leq u_i(y;A)$. Since both the $\min$ and $\max$ operations in~\eqref{eq:medianselection} are nondecreasing in each argument, we obtain
\[
 \begin{aligned}
 \Med_i(x;A)
 &\leq\max\{\ell_i(y;A),\min\{y_i,u_i(y;A)\}\}\\
 &=\Med_i(y;A).
 \end{aligned}
\]
Combining this with the WA case establishes $F_\zero(x)\leq F_\zero(y)$. Since $\gamma_i\in(0,1)$ for any $i\in\V$, we have
\[
\begin{aligned}
F_{\gamma,i}(x)
&=\gamma_i x_i+(1-\gamma_i)F_{\zero,i}(x)\\
&\leq\gamma_i y_i+(1-\gamma_i)F_{\zero,i}(y)\\
&=F_{\gamma,i}(y).
\end{aligned}
\]
Hence $F_\gamma(x)\leq F_\gamma(y)$, proving that both $F_\zero$ and $F_\gamma$ are monotone.

For any $c\in\R$, we have $C_i(t;x+c\one)=C_i(t-c;x)$; hence
\[
 \begin{aligned}
 \ell_i(x+c\one;A)&=\ell_i(x;A)+c,\\
 u_i(x+c\one;A)&=u_i(x;A)+c,\\
 \Med_i(x+c\one;A)&=\Med_i(x;A)+c.
 \end{aligned}
\]
Combining these identities with the WA case yields $F_\zero(x+c\one)=F_\zero(x)+c\one$. By~\eqref{eq:system}, for any $i\in\V$, we have
\[
\begin{aligned}
F_{\gamma,i}(x+c\one)
&=\gamma_i(x_i+c)+(1-\gamma_i)(F_{\zero,i}(x)+c)\\
&=F_{\gamma,i}(x)+c.
\end{aligned}
\]
Hence $F_\gamma(x+c\one)=F_\gamma(x)+c\one$, establishing the additive homogeneity of both $F_\zero$ and $F_\gamma$.
In addition, for any $c\in\R$, both the WA and WM responses equal $c$ at $x=c\one$, yielding $F_\zero(c\one)=c\one$. By~\eqref{eq:system}, we also have $F_\gamma(c\one)=c\one$. Thus both maps fix every consensus vector.

Finally, for any $x\in\R^n$, since $F_\gamma$ is monotone and fixes every consensus vector, applying it to $m(x)\one\leq x\leq M(x)\one$ yields
\[
 m(x)\one=F_\gamma(m(x)\one)\leq F_\gamma(x)
 \leq F_\gamma(M(x)\one)=M(x)\one.
\]
Taking the minimum and maximum over $i\in\V$ in the inequalities above yields $m(x)\leq m(F_\gamma(x))$ and $M(F_\gamma(x))\leq M(x)$, respectively. Since $m(F_\gamma(x))\leq M(F_\gamma(x))$ by definition, we obtain~\eqref{eq:hull}. With $x=x(k)$, this becomes
\[
 m(x(k))\leq m(x(k+1))\leq M(x(k+1))\leq M(x(k)).
\]
Iteration keeps every component in $[m(x(0)),M(x(0))]$ and proves boundedness.
\end{proof}

\subsection{Equilibria and type-dependent cohesion}
We first establish the componentwise equilibrium conditions for system~\eqref{eq:system}, and then introduce opinion level sets to derive a necessary and sufficient characterization of its equilibria involving graph-theoretic conditions and WA balance equations.

\begin{theorem}[Equilibrium conditions]\label{thm:nodeequilibrium}
Consider system~\eqref{eq:system} with a fixed nonnegative row-stochastic adjacency matrix $A$, a fixed WA--WM configuration $(\Vwa,\Vwm)$, and $\gamma\in(0,1)^n$. An opinion vector $x^\star\in\R^n$ is an equilibrium if and only if
\begin{align}
 x_i^\star&=\sum_{j\in\V}a_{ij}x_j^\star,
 &&i\in\Vwa,\label{eq:waequilibrium}\\
 \sum_{j:x_j^\star<x_i^\star}a_{ij}&\leq\tfrac12,
 \quad
 \sum_{j:x_j^\star>x_i^\star}a_{ij}\leq\tfrac12,
 &&i\in\Vwm.\label{eq:wmequilibrium}
\end{align}
In particular, the equilibrium set is independent of the values of $\gamma_i\in(0,1)$ and contains every consensus vector.
\end{theorem}
\begin{proof}
Since $1-\gamma_i>0$, the equality $F_{\gamma,i}(x^\star)=x_i^\star$ is equivalent to $F_{\zero,i}(x^\star)=x_i^\star$. For WA individuals this is~\eqref{eq:waequilibrium}. For WM individuals, the projection in~\eqref{eq:medianselection} fixes $x_i^\star$ exactly when $x_i^\star\in\mathcal M_i(x^\star;A)$. Lemma~\ref{lem:mediancharacterization} gives~\eqref{eq:wmequilibrium}. Neither condition contains $\gamma_i$, and both hold for consensus vectors.
\end{proof}

We now define type-dependent cohesive sets to reflect the different requirements of WA and WM individuals, building on the cohesive-set notion in~\cite{mei2024}.

\begin{definition}[Partial and type-dependent cohesive sets]\label{def:cohesion}
For the influence network with adjacency matrix $A$ and WA--WM configuration $(\Vwa,\Vwm)$, a nonempty set $S\subseteq\V$ is \emph{partially cohesive} if
\begin{equation}\label{eq:partial}
 a_i(S)\geq\tfrac12\qquad\text{for any }i\in S\cap\Vwm.
\end{equation}
It is \emph{type-dependent cohesive} if, in addition,
\begin{equation}\label{eq:special}
 a_i(S)=1\qquad\text{for any }i\in S\cap\Vwa.
\end{equation}
Thus partial cohesion imposes no constraint on the WA members. Let $\mathscr C(A,\Vwa)$ denote the family of nonempty type-dependent cohesive sets.
\end{definition}

Under~\eqref{eq:medianselection}, a WM individual retains an extreme opinion if it assigns at least half of its weight to individuals sharing that opinion. To retain that same extreme opinion under weighted averaging, a WA individual must assign all its weight to individuals holding that extreme opinion. This full-weight requirement concerns preservation of an extreme opinion, not every WA equilibrium equation. In graph-theoretic terms, \eqref{eq:special} excludes outgoing edges from WA members of a type-dependent cohesive set $S$ to $\V\setminus S$. Thus, in the all-WA model, a type-dependent cohesive set is precisely a nonempty set with no outgoing edges to its complement.

For an opinion vector $x$ and a threshold $\theta\in\R$, define
\begin{equation}\label{eq:levelsets}
 \begin{gathered}
 \V_{\geq\theta}(x)=\{i\in\V:x_i\geq\theta\},\quad
 \V_{\leq\theta}(x)=\{i\in\V:x_i\leq\theta\},\\
 \V_{\max}(x)=\argmax_{i\in\V}x_i,\quad
 \V_{\min}(x)=\argmin_{i\in\V}x_i,\\
 \V_{\mathrm{int}}(x)=\{i\in\V:m(x)<x_i<M(x)\}.
 \end{gathered}
\end{equation}
We refer to the sets in~\eqref{eq:levelsets} collectively as \emph{opinion level sets}, with $\V_{\geq\theta}(x)$ and $\V_{\leq\theta}(x)$ denoting the upper and lower opinion level sets of $x$ at threshold $\theta$, respectively. Unlike $\mathscr C(A,\Vwa)$, these sets depend on the opinion vector. Every consensus vector is already an equilibrium by Theorem~\ref{thm:nodeequilibrium}, so the next result focuses on nonconstant opinion vectors.

\begin{theorem}[Equilibrium characterization via opinion level sets]\label{thm:levelequilibrium}
Consider system~\eqref{eq:system} with a fixed nonnegative row-stochastic adjacency matrix $A$, a fixed WA--WM configuration $(\Vwa,\Vwm)$, and $\gamma\in(0,1)^n$, and let $x^\star\in\R^n$ be a nonconstant opinion vector. With the sets in~\eqref{eq:levelsets} and cohesion defined by Definition~\ref{def:cohesion}, $x^\star$ is an equilibrium if and only if all three conditions hold:
\begin{enumerate}
\item The sets $\V_{\max}(x^\star)$ and $\V_{\min}(x^\star)$ themselves belong to $\mathscr C(A,\Vwa)$.
\item For any $i\in\Vwa\cap\V_{\mathrm{int}}(x^\star)$, the WA equilibrium condition~\eqref{eq:waequilibrium} holds.
\item For any $i\in\Vwm\cap\V_{\mathrm{int}}(x^\star)$, both $\V_{\geq x_i^\star}(x^\star)$ and $\V_{\leq x_i^\star}(x^\star)$ are partially cohesive.
\end{enumerate}
\end{theorem}
\begin{proof}
\emph{Necessity.} Suppose that $x^\star$ is an equilibrium, and put $M^\star=M(x^\star)$ and $m^\star=m(x^\star)$. For any WA individual $i\in\Vwa\cap\V_{\max}(x^\star)$, using $\sum_{j\in\V}a_{ij}=1$ and~\eqref{eq:waequilibrium}, we obtain
\begin{equation}\label{eq:extremebalance}
\begin{aligned}
0&=M^\star-\sum_{j\in\V}a_{ij}x_j^\star\\
&=\sum_{j\in\V}a_{ij}(M^\star-x_j^\star)\\
&=\sum_{j\notin\V_{\max}(x^\star)}a_{ij}(M^\star-x_j^\star).
\end{aligned}
\end{equation}
The last equality follows from $x_j^\star=M^\star$ for any $j\in\V_{\max}(x^\star)$. In the last sum of~\eqref{eq:extremebalance}, each term $a_{ij}(M^\star-x_j^\star)$ is nonnegative, since $a_{ij}\geq0$ and the opinion difference $M^\star-x_j^\star$ is strictly positive for any $j\notin\V_{\max}(x^\star)$. Since this sum equals zero, we must have $a_{ij}=0$ for any $j\notin\V_{\max}(x^\star)$. Together with $\sum_{j\in\V}a_{ij}=1$, this yields $a_i(\V_{\max}(x^\star))=1$. For any WM individual $i\in\Vwm\cap\V_{\max}(x^\star)$, ~\eqref{eq:wmequilibrium} gives
\begin{equation}\label{eq:extrememedian}
 a_i(\V_{\max}(x^\star))
 =1-\sum_{j:x_j^\star<M^\star}a_{ij}\geq\tfrac12.
\end{equation}
The argument at the minimum is analogous. Thus, by Definition~\ref{def:cohesion}, both extreme sets are type-dependent cohesive. 

Condition~2 holds vacuously if $\Vwa\cap\V_{\mathrm{int}}(x^\star)=\varnothing$, and condition~3 holds vacuously if $\Vwm\cap\V_{\mathrm{int}}(x^\star)=\varnothing$. We therefore consider the corresponding nonempty cases. For any WA individual $i\in\Vwa\cap\V_{\mathrm{int}}(x^\star)$, the equilibrium condition is precisely condition~2.

Take any $\theta\in\R$ for which $\V_{\geq\theta}(x^\star)$ is nonempty, and any $j\in\Vwm\cap\V_{\geq\theta}(x^\star)$. By~\eqref{eq:wmequilibrium},
\begin{equation}\label{eq:upperlevelmass}
 \begin{aligned}
 a_j(\V_{\geq\theta}(x^\star))
 &\geq a_j(\V_{\geq x_j^\star}(x^\star))\\
 &=1-\sum_{h:x_h^\star<x_j^\star}a_{jh}\geq\tfrac12.
 \end{aligned}
\end{equation}
Thus every nonempty upper opinion level set is partially cohesive. Likewise, if $\V_{\leq\theta}(x^\star)\ne\varnothing$, each WM member $j$ satisfies
\[
 a_j(\V_{\leq\theta}(x^\star))
 \geq1-\sum_{h:x_h^\star>x_j^\star}a_{jh}\geq\tfrac12.
\]
Each upper and lower opinion level set specified in condition~3 contains $i$ and hence is nonempty. This establishes necessity even when an opinion level set contains only WA or only WM individuals.

\emph{Sufficiency.} Assume conditions~1--3. For any WA individual $i\in\Vwa\cap\V_{\max}(x^\star)$, condition~1 gives $a_i(\V_{\max}(x^\star))=1$, implying $\sum_{j\in\V}a_{ij}x_j^\star=M^\star=x_i^\star$, so its response equals its current opinion $x_i^\star$. The case $i\in\Vwa\cap\V_{\min}(x^\star)$ is analogous. For any WM individual $i\in\Vwm\cap\V_{\max}(x^\star)$, condition~1 gives $a_i(\V_{\max}(x^\star))\geq\tfrac12$, implying $\sum_{j:x_j^\star<M^\star}a_{ij}\leq\tfrac12$. Since $\sum_{j:x_j^\star>M^\star}a_{ij}=0$, both inequalities in~\eqref{eq:wmequilibrium} are satisfied, so its response equals its current opinion $x_i^\star$ as well. The case $i\in\Vwm\cap\V_{\min}(x^\star)$ is analogous. Condition~2 gives the equilibrium equation for any WA individual $i\in\Vwa\cap\V_{\mathrm{int}}(x^\star)$. For any WM individual $i\in\Vwm\cap\V_{\mathrm{int}}(x^\star)$, condition~3 yields
\begin{equation}\label{eq:levelcomplements}
\begin{aligned}
\sum_{j:x_j^\star<x_i^\star}a_{ij}
&=1-a_i(\V_{\geq x_i^\star}(x^\star))\leq\tfrac12,\\
\sum_{j:x_j^\star>x_i^\star}a_{ij}
&=1-a_i(\V_{\leq x_i^\star}(x^\star))\leq\tfrac12.
\end{aligned}
\end{equation}
By Theorem~\ref{thm:nodeequilibrium}, the inequalities above imply that the response of the WM individual $i\in\Vwm\cap\V_{\mathrm{int}}(x^\star)$ equals its current opinion $x_i^\star$.
The equilibrium conditions therefore hold for every individual, completing the proof.
\end{proof}

\begin{remark}\label{rem:extremes}
Conditions~1 and~3 impose graph-theoretic requirements on the specified opinion level sets. In particular, the extreme sets themselves must be type-dependent cohesive; containment in larger type-dependent cohesive sets is insufficient. Condition~2 additionally requires the numerical balance equation to hold for every WA individual whose opinion lies strictly between the minimum and maximum opinions. Unlike the consensus criterion below, the equilibrium characterization in Theorem~\ref{thm:levelequilibrium} depends on the opinion vector as well as the adjacency matrix and WA--WM configuration.
\end{remark}

\subsection{Stability and convergence to equilibrium}
Having characterized equilibria, we now investigate their stability and whether every solution converges to an equilibrium. We measure distances in the maximum norm $\|x\|_\infty=\max_{i\in\V}|x_i|$. A map is nonexpansive in this norm if applying it cannot increase the distance between two opinion vectors. This property is studied for weighted-median maps in~\cite[Theorem~8]{han2025}. We verify it here for $F_\zero$ and $F_\gamma$.

\begin{lemma}[Nonexpansiveness]\label{lem:nonexpansive}
For the influence network, WA--WM configuration, with $\gamma\in(0,1)^n$, the maps $F_\zero$ and $F_\gamma$ in~\eqref{eq:response}--\eqref{eq:system} satisfy
\begin{equation}\label{eq:nonexpansive}
 \begin{aligned}
 \|F_\zero(x)-F_\zero(y)\|_\infty&\leq\|x-y\|_\infty,\\
 \|F_\gamma(x)-F_\gamma(y)\|_\infty&\leq\|x-y\|_\infty
 \end{aligned}
 \qquad \text{for any } x,y\in\R^n.
\end{equation}
\end{lemma}
\begin{proof}
Set $\delta=\|x-y\|_\infty$. Then $y-\delta\one\leq x\leq y+\delta\one$. By the monotonicity and additive homogeneity of both maps established in Lemma~\ref{lem:order}, we obtain
\[
 \begin{aligned}
 F_\zero(y)-\delta\one&\leq F_\zero(x)\leq F_\zero(y)+\delta\one,\\
 F_\gamma(y)-\delta\one&\leq F_\gamma(x)\leq F_\gamma(y)+\delta\one,
 \end{aligned}
\]
which is~\eqref{eq:nonexpansive}.
\end{proof}

\begin{corollary}[Stability of equilibria]\label{thm:stability}
Consider system~\eqref{eq:system} with a fixed nonnegative row-stochastic adjacency matrix $A$, a fixed WA--WM configuration $(\Vwa,\Vwm)$, and $\gamma\in(0,1)^n$. For any equilibrium $x^\star$ and initial condition $x(0)\in\R^n$,
\begin{equation}\label{eq:stability}
\|x(k)-x^\star\|_\infty
\leq\|x(0)-x^\star\|_\infty,
\qquad \text{for any } k\in\mathbb N_0.
\end{equation}
Hence every equilibrium is Lyapunov stable, but this does not imply asymptotic stability.
\end{corollary}
\begin{proof}
Since $F_\gamma(x^\star)=x^\star$, Lemma~\ref{lem:nonexpansive} gives, for any $k\in\mathbb N_0$,
\[
 \begin{aligned}
 \|x(k+1)-x^\star\|_\infty
 &=\|F_\gamma(x(k))-F_\gamma(x^\star)\|_\infty\\
 &\leq\|x(k)-x^\star\|_\infty\\
 &\leq\cdots\leq\|x(0)-x^\star\|_\infty.
 \end{aligned}
\]
For any $\varepsilon>0$, choose $\delta=\varepsilon$. Then $\|x(0)-x^\star\|_\infty<\delta$ implies $\|x(k)-x^\star\|_\infty<\varepsilon$ for all $k\in\mathbb N_0$, which is Lyapunov stability. 
\end{proof}

Nonexpansiveness alone does not imply convergence. To establish convergence, we introduce the following finite-dimensional form of Krasnosel'skii--Mann iteration. We use the result of Bravo et al.~\cite[Theorem~1]{bravo2019} to show that the fixed-point residual tends to zero, i.e., $\|T(z(k))-z(k)\|\to0$ (asymptotic regularity).

\begin{lemma}[Finite-dimensional averaged iteration~\cite{bravo2019}]\label{lem:km}
Let $T:\R^n\to\R^n$ be nonexpansive in a norm $\|\cdot\|$, with $\Fix(T)\neq\varnothing$. For any constant $\alpha\in(0,1)$ and $z(0)\in\R^n$, the iteration
\begin{equation}\label{eq:kmiteration}
 z(k+1)=(1-\alpha)z(k)+\alpha T(z(k))
\end{equation}
converges in norm to a point in $\Fix(T)$.
\end{lemma}
\begin{proof}
For any $p\in\Fix(T)$, the triangle inequality and the nonexpansiveness of $T$ imply that
\begin{equation}\label{eq:km_distance}
\begin{aligned}
\|z(k+1)-p\|
&\leq(1-\alpha)\|z(k)-p\|\\
&\quad+\alpha\|T(z(k))-T(p)\|\\
&\leq\|z(k)-p\|\leq\|z(0)-p\|.
\end{aligned}
\end{equation}
Put $r=\|z(0)-p\|$ and $C=\{z:\|z-p\|\leq r\}$. This is a closed bounded convex set, and $T(C)\subseteq C$ because $\|T(z)-p\|\leq\|z-p\|$. Thus, the averaged iterates $z(k)$ in~\eqref{eq:kmiteration} remain in $C$ for all $k\in\mathbb N_0$. In particular,
\[
 \|T(z(k))-z(0)\|
 \leq\|T(z(k))-p\|+\|p-z(0)\|\leq2r.
\]
This verifies hypothesis (h0) of~\cite[Theorem~1]{bravo2019}. 
To account for cases where $T(z(k-1))$ cannot be computed exactly, Bravo et al.~\cite{bravo2019} consider an approximation $\widetilde T(z(k-1))$ and define the error $e_k=\widetilde T(z(k-1))-T(z(k-1))$, with a nonnegative upper bound $\varepsilon_k\geq\|e_k\|$. Since our iteration~\eqref{eq:kmiteration} is exact, we have $e_k=\zero$ and may choose $\varepsilon_k=0$ for all $k\in\mathbb N_+$. Thus both error conditions, $\varepsilon_k\to0$ and $\sum_{k=1}^{\infty}\alpha\varepsilon_k<\infty$, are satisfied.
Moreover,
\[
 \sum_{j=1}^{k}\alpha(1-\alpha)=k\alpha(1-\alpha)\longrightarrow\infty.
\]
Therefore, Theorem~1 of Bravo et al.~\cite{bravo2019}
implies that $\|T(z(k))-z(k)\|\to0$ as $k\to\infty$,
which also implies $\|T(z(k_j))-z(k_j)\|\to0$
as $j\to\infty$ for any subsequence $\{z(k_j)\}$.

Since $\|z(k)-p\|\leq r$ for all $k\in\mathbb N_0$, the triangle inequality gives
\[
\|z(k)\|\leq\|z(k)-p\|+\|p\|\leq r+\|p\|.
\]
Thus $\{z(k)\}$ is bounded in $\R^n$. By the Bolzano--Weierstrass theorem, there exists a convergent subsequence $z(k_j)\to z^\star$ as $j\to\infty$.
By the nonexpansiveness of $T$ and the triangle inequality, we have
\[
 \begin{aligned}
 \|T(z^\star)-z^\star\|
 &\leq\|T(z^\star)-T(z(k_j))\|\\
 &\quad+\|T(z(k_j))-z(k_j)\|+\|z(k_j)-z^\star\|\\
 &\leq2\|z(k_j)-z^\star\|+\|T(z(k_j))-z(k_j)\|
 \longrightarrow0.
 \end{aligned}
\]
Hence $z^\star\in\Fix(T)$, and setting $p=z^\star$ in~\eqref{eq:km_distance} shows that $\|z(k)-z^\star\|$ is nonincreasing. For $k\geq k_j$, it is at most $\|z(k_j)-z^\star\|$, which converges to zero as $j\to\infty$. Hence $\|z(k)-z^\star\|\to0$ as $k\to\infty$, implying $z(k)\to z^\star$. 
\end{proof}

\begin{theorem}[Global convergence]\label{thm:convergence}
Consider system~\eqref{eq:system} with a fixed nonnegative row-stochastic adjacency matrix $A$, a fixed WA--WM configuration $(\Vwa,\Vwm)$, and $\gamma\in(0,1)^n$. For any $x(0)\in\R^n$, there exists an equilibrium $x^\star$ such that $x(k)\to x^\star$. The limit $x^\star$ satisfies $m(x(0))\leq x_i^\star\leq M(x(0))$ for any $i\in\V$. As an equilibrium, $x^\star$ also satisfies the conditions in Theorem~\ref{thm:nodeequilibrium} and, if it is a dissensus equilibrium, the conditions in Theorem~\ref{thm:levelequilibrium}.
\end{theorem}
\begin{proof}

To handle the heterogeneous inertia coefficients $\gamma=(\gamma_1,\ldots,\gamma_n)^\top\in(0,1)^n$, let $\alpha=\max_i(1-\gamma_i)\in(0,1)$ and $\eta_i=(1-\gamma_i)/\alpha\in(0,1]$. By~\eqref{eq:system} and $1-\gamma_i=\alpha\eta_i$, we have
\begin{equation}\label{eq:componentwiseexpansion}
\begin{aligned}
F_{\gamma,i}(x)
&=\gamma_i x_i+(1-\gamma_i)F_{\zero,i}(x)\\
&=(1-\alpha)x_i\\
&\quad+\alpha\big[(1-\eta_i)x_i+\eta_i F_{\zero,i}(x)\big].
\end{aligned}
\end{equation}
Comparing this expression with the averaged iteration in Lemma~\ref{lem:km}, we define an auxiliary map $G:\R^n\to\R^n$ by
\begin{equation}\label{eq:auxiliarymap}
G_i(x)=(1-\eta_i)x_i+\eta_i F_{\zero,i}(x),
\qquad \text{for any } i\in\V.
\end{equation}
For $\delta=\|x-y\|_\infty$, Lemma~\ref{lem:nonexpansive} yields
\begin{equation}\label{eq:auxiliarynonexpansive}
 \begin{aligned}
 |G_i(x)-G_i(y)|
 &\leq(1-\eta_i)|x_i-y_i|\\
 &\quad+\eta_i|F_{\zero,i}(x)-F_{\zero,i}(y)|
 \leq\delta.
 \end{aligned}
\end{equation}
Thus $G$ is nonexpansive in the maximum norm. The definitions of $F_\gamma$ and $G$ give
\[
 \begin{aligned}
 F_{\gamma,i}(x)-x_i&=(1-\gamma_i)(F_{\zero,i}(x)-x_i),\\
 G_i(x)-x_i&=\alpha^{-1}(1-\gamma_i)(F_{\zero,i}(x)-x_i).
 \end{aligned}
\]
Each factor $1-\gamma_i$ is positive, and $\alpha>0$. Each of $F_\gamma(x)=x$ and $G(x)=x$ is therefore equivalent to $F_\zero(x)=x$. Since $F_\zero$ fixes consensus vectors by Lemma~\ref{lem:order},
\begin{equation}\label{eq:fixedpointequivalence}
\Fix(G)=\Fix(F_\zero)=\Fix(F_\gamma)\supseteq\{c\one:c\in\R\}.
\end{equation}
Let $I:\R^n\to\R^n$ denote the identity map, $I(x)=x$. Substituting~\eqref{eq:auxiliarymap} into~\eqref{eq:componentwiseexpansion} gives
\begin{equation}\label{eq:averagedform}
F_\gamma=(1-\alpha)I+\alpha G.
\end{equation}
The map $G$ is nonexpansive on $\R^n$, its fixed-point set is nonempty by~\eqref{eq:fixedpointequivalence}, and $\alpha\in(0,1)$ is constant. These verify all hypotheses of Lemma~\ref{lem:km} for the original solution with $z(0)=x(0)$. The solution converges to some $x^\star\in\Fix(G)=\Fix(F_\gamma)$, which is therefore an equilibrium of system~\eqref{eq:system}.  Lemma~\ref{lem:order} keeps each component in the initial opinion interval. The equilibrium characterizations in Theorems~\ref{thm:nodeequilibrium} and~\ref{thm:levelequilibrium} establish the remaining claims.
\end{proof}

If the assumption $\gamma_i>0$ in Theorem~\ref{thm:convergence} is relaxed to allow $\gamma_i=0$ for some $i\in\V$, global convergence is no longer guaranteed. For example, consider two WA individuals with $a_{12}=a_{21}=1$ and $\gamma_1=\gamma_2=0$; if their initial opinions differ, the two opinions are exchanged at each update. This periodic example does not imply that every network with a zero inertia coefficient fails to converge. Strictly positive inertia is a sufficient assumption for the general convergence guarantee, not a necessary condition for each particular network. The upper bound $\gamma_i<1$ ensures that every individual gives positive weight to its response; without it, the fixed-point equivalence in~\eqref{eq:fixedpointequivalence} need not hold.

To facilitate the analysis of the numerical simulations in Section~\ref{sec:numerics}, define the equilibrium residual $R(x)=\|F_\zero(x)-x\|_\infty$, and let $\gamma_{\min}=\min_{i\in\V}\gamma_i$ and $\gamma_{\max}=\max_{i\in\V}\gamma_i$. Since $F_{\gamma,i}(x)-x_i=(1-\gamma_i)(F_{\zero,i}(x)-x_i)$, bounding each positive multiplier between $1-\gamma_{\max}$ and $1-\gamma_{\min}$ and taking the maximum gives
\begin{equation}\label{eq:residualcomparison}
 (1-\gamma_{\max})R(x)\leq\|F_\gamma(x)-x\|_\infty
 \leq(1-\gamma_{\min})R(x).
\end{equation}
By Theorem~\ref{thm:convergence}, the sequence $x(k)$ converges to an equilibrium $x^\star$. Hence
\[
 \begin{aligned}
 \|F_\gamma(x(k))-x(k)\|_\infty
 &=\|x(k+1)-x(k)\|_\infty\\
 &\leq\|x(k+1)-x^\star\|_\infty+\|x(k)-x^\star\|_\infty
 \longrightarrow0.
 \end{aligned}
\]
Combining this with~\eqref{eq:residualcomparison} and $1-\gamma_{\max}>0$ gives $R(x(k))\to0$. $R(x)$ measures the residual of the equilibrium equation $F_\zero(x)=x$, rather than the opinion diameter $D(x)$ defined in~\eqref{eq:diameter} or the distance to the limiting equilibrium. A small residual does not imply consensus: if $x(k)\to x^\star$ for a dissensus equilibrium $x^\star$, then $R(x(k))\to0$, while $D(x(k))\to D(x^\star)>0$. Unlike the opinion update map $F_\gamma$ in~\eqref{eq:system}, the residual $R(x)$ does not involve the inertia coefficients $\gamma_i$.

\subsection{Conditions for consensus and dissensus}
Having established global convergence, we next derive necessary and sufficient graph-theoretic conditions for system~\eqref{eq:system} to reach consensus from any initial condition and investigate when dissensus may occur. The following lemma characterizes type-dependent cohesive sets through the positive invariance of their initial opinion intervals.

\begin{lemma}[\textit{Invariance for type-dependent cohesive sets}]\label{lem:groupinterval}
Consider system~\eqref{eq:system} with a fixed nonnegative row-stochastic adjacency matrix $A$, a fixed WA--WM configuration $(\Vwa,\Vwm)$, and $\gamma\in(0,1)^n$, and let $\varnothing\ne S\subseteq\V$. With $\mathscr C(A,\Vwa)$ defined in Definition~\ref{def:cohesion}, $S\in\mathscr C(A,\Vwa)$ if and only if, for any $x\in\R^n$ and $i\in S$,
\begin{equation}\label{eq:grouponeupdate}
 \min_{j\in S}x_j\leq F_{\gamma,i}(x)\leq\max_{j\in S}x_j.
\end{equation}
Consequently, for such a set and every solution,
\begin{equation}\label{eq:groupinterval}
 \min_{j\in S}x_j(k)\leq x_i(k+1)\leq\max_{j\in S}x_j(k).
\end{equation}
In particular, the following initial opinion interval is positively invariant for all members of $S$:
\[
\left[\min_{j\in S}x_j(0),\,\max_{j\in S}x_j(0)\right].
\]
\end{lemma}
\begin{proof}
\emph{Necessity.} Write $a=\min_{j\in S}x_j$ and $b=\max_{j\in S}x_j$. Since $S$ is type-dependent cohesive, every WA member of $S$ assigns all its weight to $S$, giving $F_{\zero,i}(x)\in[a,b]$ for any $i\in S\cap\Vwa$. For any $i\in S\cap\Vwm$, at least half of its weight is assigned to $[a,b]$. Consequently, its weighted-median interval $[\ell_i(x;A),u_i(x;A)]$ intersects $[a,b]$, i.e., $\ell_i(x;A)\le b$ and $u_i(x;A)\ge a$. The weighted-median selection rule~\eqref{eq:medianselection} ensures that the weighted-median response remains in $[a,b]$, i.e., $F_{\zero,i}(x)\in[a,b]$, even when $a_i(S)=1/2$. The convex combination of $x_i$ and $F_{\zero,i}(x)$, with weights $\gamma_i$ and $1-\gamma_i$, also lies in $[a,b]$, proving~\eqref{eq:grouponeupdate}. Substituting $x=x(k)$ into~\eqref{eq:grouponeupdate} for any $k\in\mathbb N_0$ yields~\eqref{eq:groupinterval}, and iterating this inequality establishes the positive invariance of the initial opinion interval of $S$.

\emph{Sufficiency.} We prove the contrapositive: if $S\notin\mathscr C(A,\Vwa)$, there exist $x\in\R^n$ and $i\in S$ for which~\eqref{eq:grouponeupdate} fails. By Definition~\ref{def:cohesion} and $a_i(S)\leq1$, such a set has at least one WA member with $a_i(S)<1$ or at least one WM member with $a_i(S)<1/2$. Consider the opinion vector
\begin{equation}\label{eq:binaryteststate}
 x_j=0\quad(j\in S),\qquad x_j=1\quad(j\notin S).
\end{equation}
If a WA member $i\in S\cap\Vwa$ violates~\eqref{eq:special}, then
\[
 F_{\gamma,i}(x)=(1-\gamma_i)\sum_{j\notin S}a_{ij}
       =(1-\gamma_i)(1-a_i(S))>0.
\]
If a WM member $i\in S\cap\Vwm$ violates~\eqref{eq:partial}, the opinion $1$ carries strictly more than half its weight. Its unique weighted median is $1$, hence $F_{\gamma,i}(x)=1-\gamma_i>0$. For this opinion vector, the original group interval is $[\min_{j\in S}x_j,\max_{j\in S}x_j]=[0,0]$. The violating member has $F_{\gamma,i}(x)>0$ in either case. This negates the assertion that~\eqref{eq:grouponeupdate} holds for all $x\in\R^n$ and $i\in S$, establishing sufficiency.
\end{proof}

Lemma~\ref{lem:groupinterval} extends the positive-invariance result for cohesive sets in continuous-time WM dynamics~\cite[Proposition~10]{han2025} to the discrete-time mixed WA--WM system~\eqref{eq:system}. The weighted-median selection rule~\eqref{eq:medianselection} keeps the weighted-median response inside the current opinion interval of a type-dependent cohesive set even at a half-weight tie; an arbitrary choice from the weighted-median interval need not do so.

\begin{theorem}[Graph-theoretic consensus criterion]\label{thm:consensus}
Consider system~\eqref{eq:system} with a fixed nonnegative row-stochastic adjacency matrix $A$, a fixed WA--WM configuration $(\Vwa,\Vwm)$, and $\gamma\in(0,1)^n$, with the type-dependent cohesive-set family $\mathscr C(A,\Vwa)$ defined in Definition~\ref{def:cohesion}. The following statements are equivalent:
\begin{enumerate}
\item For any initial opinion vector $x(0)\in\R^n$, the solution converges to a consensus vector.
\item Any equilibrium of~\eqref{eq:system} is a consensus equilibrium.
\item There are no nonempty disjoint sets $S_-,S_+\in\mathscr C(A,\Vwa)$.
\end{enumerate}
If condition~3 fails, there is a nonempty open, bounded set of initial conditions, of positive Lebesgue measure in $\R^n$, for which the solutions converge to dissensus equilibria.
\end{theorem}
\begin{proof}
Statement~1 implies statement~2. Indeed, for any equilibrium $x^\star$, choosing $x(0)=x^\star$ gives $x(k)=x^\star$ for all $k\in\mathbb N_0$. Since statement~1 requires every solution to converge to consensus, $x^\star$ must be a consensus vector. Thus any equilibrium is a consensus equilibrium, proving statement~2. Conversely, by Theorem~\ref{thm:convergence}, every solution converges to an equilibrium, which must be a consensus equilibrium by statement~2. Thus statement~1 holds.

If statement~3 holds, a dissensus equilibrium would, by Theorem~\ref{thm:levelequilibrium}, have two nonempty disjoint type-dependent cohesive extreme sets. This contradiction proves statement~2.

Suppose that statement~3 fails, and choose disjoint $S_-,S_+\in\mathscr C(A,\Vwa)$. Choose $a,b,c,d\in\R$ such that $a<b<c<d$, and define $\mathcal B\subset\R^n$ as the set of initial conditions $x(0)$ satisfying
\begin{equation}\label{eq:dissensusbox}
\begin{cases}
x_i(0)\in(a,b), & i\in S_-,\\
x_i(0)\in(c,d), & i\in S_+,\\
x_i(0)\in(a,d), & i\in\V\setminus(S_-\cup S_+).
\end{cases}
\end{equation}
By Lemma~\ref{lem:groupinterval}, for any $k\in\mathbb N_0$,
\begin{equation}\label{eq:persistentgap}
 \begin{aligned}
 D(x(k))
 &\geq\min_{i\in S_+}x_i(k)-\max_{i\in S_-}x_i(k)\\
 &\geq\min_{i\in S_+}x_i(0)-\max_{i\in S_-}x_i(0)>c-b.
 \end{aligned}
\end{equation}
By Theorem~\ref{thm:convergence}, every solution converges to an equilibrium $x^\star$. For any $x(0)\in\mathcal B$, taking the limit in the inequality above yields $D(x^\star)>c-b>0$, so the corresponding solution converges to a dissensus equilibrium. The set $\mathcal B$ is nonempty, open, and bounded, and hence has positive Lebesgue measure. Statement~1 therefore fails for these initial conditions, and the resulting dissensus equilibria disprove statement~2, establishing both claims.
\end{proof}

Theorem~\ref{thm:consensus} also applies when all individuals are either WM or WA. Let $S\subseteq\V$ be any nonempty set. In the all-WM case, $S$ is type-dependent cohesive if and only if $a_i(S)\geq1/2$ for any $i\in S$, which coincides with the cohesive-set condition in~\cite{mei2024}, although that model uses asynchronous updates. In the all-WA case, $S$ is type-dependent cohesive if and only if $a_i(S)=1$ for any $i\in S$. For a connected network with bidirectional edges, no proper nonempty set $S$ satisfies this condition, so Theorem~\ref{thm:consensus} guarantees consensus from any initial condition in the all-WA case.

\begin{remark}[Converting WM individuals to WA]\label{rem:conversion}
Fix $A$, and let $\Vwa$ denote the original set of WA individuals. After converting some WM individuals to WA, denote the resulting WA set by $\widehat{\V}^{\wa}$, where $\Vwa\subseteq\widehat{\V}^{\wa}$. For any nonempty set $S\subseteq\V$ and any converted individual $i\in S$, the type-dependent cohesion condition changes from $a_i(S)\geq1/2$ to $a_i(S)=1$. Consequently,
\begin{equation}\label{eq:familyinclusion}
\mathscr C(A,\widehat{\V}^{\wa})\subseteq\mathscr C(A,\Vwa).
\end{equation}
Inclusion above explains how converting WM individuals to WA affects consensus:
\begin{enumerate}
\item \emph{Preservation of consensus.} Since converting WM individuals to WA cannot introduce new type-dependent cohesive sets, the converted system preserves guarantees of consensus from any initial condition whenever the original system has such guarantees.
\item \emph{Promotion of consensus.} Converting WM individuals to WA may eliminate all pairs of disjoint type-dependent cohesive sets. If it does, the converted system guarantees consensus from any initial condition even when the original system does not.
\end{enumerate}
These two implications describe the consensus-promoting role of WA individuals, although the conversion does not necessarily reduce the opinion diameter or the time required to reach a prescribed tolerance.
\end{remark}

Although the equilibrium set is independent of $\gamma$, the equilibrium reached from a given initial condition may depend on $\gamma$, as illustrated and derived for the all-WA model in Supplementary Section~S4.2. The graph-theoretic consensus criterion in Theorem~\ref{thm:consensus} is determined entirely by the influence network with adjacency matrix $A$ and the WA--WM configuration, independently of $\gamma$. 

\section{Numerical experiments}\label{sec:numerics}
Study~1 compares WA locations on a fixed network, and Study~2 varies influence weights on a two-community network. Studies~3 and~4 vary WA proportion in networks sampled from specified network ensembles; Study~4 evaluates the graph-theoretic consensus criterion and finite-horizon simulation outcomes for each sampled network.

The all-WA and all-WM models serve as baselines. Supplementary Sections~S1--S3 provide complete numerical protocols, additional counts, and further simulation results.

\subsection{Common simulation setup and numerical measures}\label{sec:simulationprotocol}
\paragraph{Simulation inputs and comparisons}
The partition $(\Vwa,\Vwm)$ specifies which individuals use each opinion-aggregation mechanism (the WA--WM configuration defined in Section~\ref{sec:model}). We write WA-$S$ for the configuration in which individuals indexed by $S$ use the WA mechanism and all others use the WM mechanism. With $A$, the configuration, and $\gamma\in(0,1)^n$ fixed, synchronous iteration of system~\eqref{eq:system} from one initial opinion vector $x(0)$ produces a numerical opinion sequence (a simulated solution).

In Studies~3 and~4, a network model, a network size, structural parameters, and an influence-weight sampling rule specify a population of possible influence networks (a network ensemble). Generating one influence network from this population gives a sampled network.

The WA proportion is $\rho=|\Vwa|/n$. For the proportion comparisons in Studies~3 and~4, a fixed ordering of the nodes (a node permutation $\pi$) determines nested WA sets $\Vwa(\rho)$: increasing $\rho$ selects more individuals in that order to use the WA mechanism, while retaining the WA individuals already selected. The remaining individuals use the WM mechanism.

Drawing a sampled network together with all other random inputs used in its comparisons---a node permutation and one or more initial opinion vectors---constitutes a complete realization. Different realizations are independent. One realization can therefore yield several simulated solutions, obtained from different WA--WM configurations or initial opinion vectors. Within each realization, comparisons across different WA proportions reuse the same sampled network, initial opinion vectors, inertia coefficients, and node permutation. Each configuration is simulated afresh from each of its initial opinion vectors; no simulated solution continues from the final opinions at another proportion. The network ensemble and initial-opinion distribution are also unchanged across these comparisons.

\paragraph{Equilibrium and consensus measures}
Consider system~\eqref{eq:system} with a fixed nonnegative row-stochastic adjacency matrix $A$, a fixed WA--WM configuration $(\Vwa,\Vwm)$, and $\gamma\in(0,1)^n$. For any initial condition $x(0)\in\R^n$ with $D(x(0))>0$ and any prescribed consensus tolerance $\varepsilon>0$, define the normalized opinion diameter $E(k)$ and the corresponding first hitting time $\tau_\varepsilon$ by
\begin{align}
&E(k)=D(x(k))/D(x(0)),\label{eq:normalizeddiameter}\\
&\tau_\varepsilon=\inf\{k\in\mathbb N_0:E(k)\leq\varepsilon\},\label{eq:attainment}
\end{align}
where $\inf\varnothing=\infty$. For an initial consensus vector, set $E(k)=0$ for all $k\in\mathbb N_0$ and $\tau_\varepsilon=0$ for any $\varepsilon\geq0$. For $D(x(0))>0$, the normalized equilibrium residual is $R(x(k))/D(x(0))$, with $R$ defined immediately before~\eqref{eq:residualcomparison}; for an initial consensus vector, set this normalized residual to zero. The normalized equilibrium residual measures the deviation from the equilibrium equation relative to the initial opinion diameter, whereas the normalized opinion diameter $E(k)$ measures the remaining spread of opinions relative to their initial spread. Since $D(x(k))$ is nonincreasing, once $E(k)\leq\varepsilon$, this inequality holds at all subsequent time steps. Numerical simulations over a finite horizon provide numerical evidence but cannot rigorously establish the theoretical results.

The simulation horizon $T$ is the maximum number of synchronous update steps considered. For the tolerance-based comparisons in Studies~2--4, we use $\varepsilon=10^{-4}$ and $T=10^5$, except $T=10^4$ in Study~2; Study~3 also reports results at shorter horizons. In Studies~2--4, we also record the first hitting times for
$\varepsilon\in\{10^{-3},10^{-4},10^{-5}\}$.
Each simulation stops when $E(k)\leq10^{-5}$,
the floating-point opinion vector no longer changes,
or the simulation horizon $T$ is reached.

\paragraph{Statistics and confidence intervals}
In Studies~3 and~4, we write the statistics as functions of $\rho$, leaving the network ensemble and other experimental settings implicit. This notation does not mean that outcomes depend on $\rho$ alone; WA locations, the sampled network, and initial opinions can also affect them. At a fixed $\rho$, let $N$ denote the number of simulated solutions used to calculate the corresponding statistics and let $\tau_\varepsilon^{(r)}(\rho)$ denote the hitting time of simulated solution $r$. The index $r$ identifies different simulated solutions. Let $\mathbf1\{P\}\in\{0,1\}$ denote the indicator function, equal to $1$ if $P$ holds and $0$ otherwise. The observed frequency of reaching the consensus tolerance is
\begin{equation}\label{eq:frequency}
 \widehat p_{\varepsilon,T}(\rho)
 =N^{-1}\sum_{r=1}^N\mathbf1\{\tau^{(r)}_\varepsilon(\rho)\leq T\}.
\end{equation}
To compare hitting times without discarding simulated solutions that do not reach the tolerance by $T$, we also report the restricted mean hitting time
\begin{equation}\label{eq:restrictedmean}
 \overline{\tau}_{\varepsilon,T}(\rho)
 =N^{-1}\sum_{r=1}^N\min\{\tau^{(r)}_\varepsilon(\rho),T\}.
\end{equation}
A simulated solution not reaching the tolerance contributes $T$; this is not an estimate of time to exact consensus. Study~2 varies the influence parameter $q$ and compares specific WA--WM configurations. For example, WA-$\{1\}$ and WA-$\{3\}$ have the same WA proportion but different WA locations, so we explicitly index the parameter setting by $(q,\Vwa)$. We therefore use $\widehat p_{\varepsilon,T}(q,\Vwa)$ for the observed frequency of reaching the consensus tolerance and $\overline{\tau}_{\varepsilon,T}(q,\Vwa)$ for the restricted mean hitting time, calculated by the same formulas with $\tau_\varepsilon^{(r)}(q,\Vwa)$ in place of $\tau_\varepsilon^{(r)}(\rho)$.

To quantify uncertainty in the observed frequency of reaching the prescribed consensus tolerance, we calculate pointwise 95\% confidence intervals at each parameter setting.

At each fixed parameter setting, the simulation outcomes are independent across the sampled initial opinion vectors in Study~2 and across realizations in Study~3, so we use Wilson intervals~\cite{wilson1927}. Let $\widehat p$ denote the observed frequency for a given parameter setting. The lower and upper bounds of the 95\% confidence interval are $c-h$ and $c+h$, respectively, where
\begin{equation}\label{eq:wilsoninterval}
\begin{aligned}
c&=\frac{\widehat p+z^2/(2N)}{1+z^2/N},\\
h&=\frac{z}{1+z^2/N}
\sqrt{\frac{\widehat p(1-\widehat p)}{N}+\frac{z^2}{4N^2}},
\end{aligned}
\end{equation}
with $z=1.96$ for a 95\% confidence interval.

In Study~4, multiple simulated solutions share the same sampled network, so their outcomes may be dependent. We therefore repeatedly sample from the previously generated networks with replacement (bootstrap resampling~\cite{efron1979}), retaining all simulation outcomes associated with each selected network (network-level bootstrap). Let $\widehat p^{*(m)}$ denote the observed frequency calculated from the $m$-th bootstrap sample. The pointwise 95\% confidence interval is
\begin{equation}\label{eq:bootstrapinterval}
[\widehat p^*_{0.025},\widehat p^*_{0.975}],
\end{equation}
where $\widehat p^*_{0.025}$ and $\widehat p^*_{0.975}$ denote the 2.5th and 97.5th percentiles, respectively, of the bootstrap frequency estimates $\{\widehat p^{*(m)}\}$ arranged in ascending order.

\paragraph{Numerical implementation}
To avoid floating-point errors when checking the half-weight conditions, we use nonnegative integer weights $B=(b_{ij})$ and obtain the row-stochastic adjacency matrix $A$ by normalizing each row, i.e., $a_{ij}=b_{ij}/\sum_{h\in\V}b_{ih}$. For any nonempty set $S\subseteq\V$, the condition $a_i(S)\geq1/2$ can then be checked exactly using integer arithmetic:
\[
2\sum_{j\in S}b_{ij}\geq\sum_{j\in\V}b_{ij}.
\]
However, floating-point rounding errors may still occur during opinion updates. Supplementary Section~S1 provides the random seeds, procedures for generating networks and initial conditions, and detailed simulation stopping criteria.

\subsection{Study 1: Effects of WA locations on consensus}\label{sec:sixnode}
\begin{figure*}[t]
 \centering
 \includegraphics[width=\textwidth]{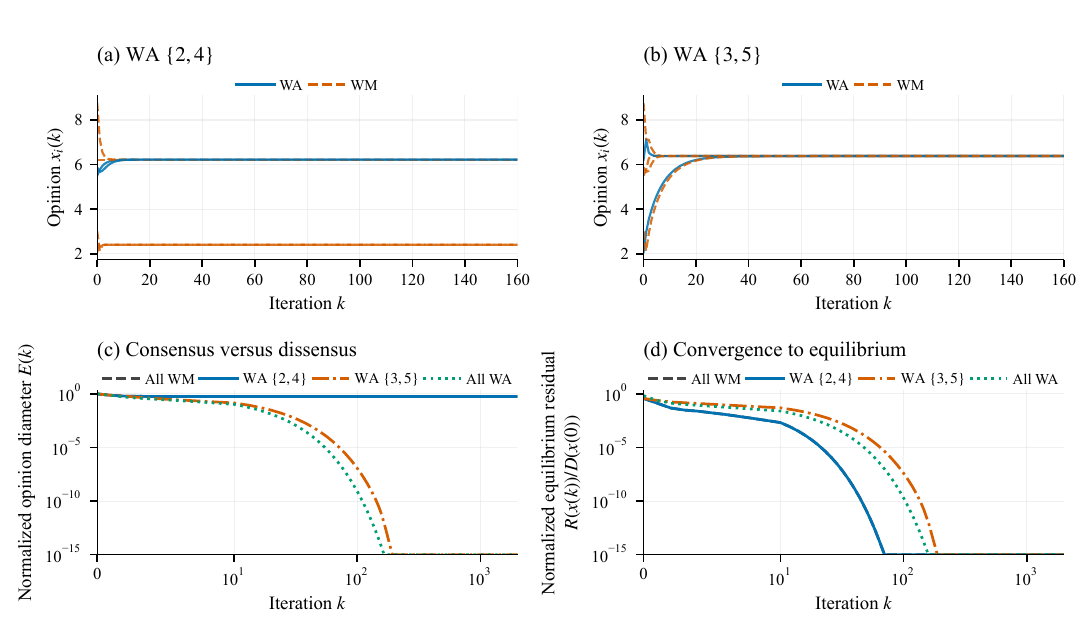}
 \caption{Study~1. Effects of WA locations on consensus in the six-node network in Fig.~\ref{fig:model}(a). The adjacency matrix is $A=B/2$ in~\eqref{eq:sixnode}. The four WA--WM configurations are all-WM, WA-$\{2,4\}$, WA-$\{3,5\}$, and all-WA. All use the same initial opinion vector $x(0)=(8.7594,\allowbreak5.5016,\allowbreak6.2248,\allowbreak5.8704,\allowbreak2.0774,\allowbreak3.0125)^\top$ and inertia coefficients $\gamma=(0.3557,\allowbreak0.6214,\allowbreak0.1884,\allowbreak0.4672,\allowbreak0.6493,\allowbreak0.0610)^\top$, with 2000 synchronous updates. (a) Opinions under WA-$\{2,4\}$. (b) Opinions under WA-$\{3,5\}$. (c) Normalized opinion diameter $E(k)$ for all four configurations. (d) Normalized equilibrium residual for all four configurations. In panels~(a,b), blue solid and orange dashed curves represent the opinions of WA and WM individuals, respectively, and only the first 160 time steps are shown. In panels~(c,d), time is linear through step 10 and logarithmic thereafter; values below $10^{-15}$ are displayed at that level.}\label{fig:mechanism}
\end{figure*}
We first examine whether the locations of WA individuals affect consensus when their number is fixed.

\paragraph{Simulation protocol}
Consider the network in Fig.~\ref{fig:model}(a), with $A=B/2$ and
\begin{equation}\label{eq:sixnode}
 B=\begin{pmatrix}
 0&0&1&0&1&0\\0&0&2&0&0&0\\1&0&0&1&0&0\\
 0&2&0&0&0&0\\0&0&1&0&0&1\\0&0&0&0&2&0
 \end{pmatrix}.
\end{equation}
Figure~\ref{fig:mechanism} compares the four WA--WM configurations on this network: all-WM, WA-$\{2,4\}$, WA-$\{3,5\}$, and all-WA. All four configurations use the same initial opinion vector and inertia coefficients specified in Fig.~\ref{fig:mechanism}, and each is simulated for 2000 synchronous updates.

\paragraph{Simulation results and analysis}
Figure~\ref{fig:mechanism} shows that WA locations do affect consensus even at the same WA proportion. Under WA-$\{2,4\}$, $\{2,3,4\}$ and $\{5,6\}$ are disjoint type-dependent cohesive sets. Under WA-$\{3,5\}$, the only type-dependent cohesive sets are $\{1,2,3,4\}$ and $\V$. The graph-theoretic criterion in Theorem~\ref{thm:consensus} therefore distinguishes these configurations through their different cohesive-set families, even though the influence network and WA proportion are unchanged. At step 2000, the opinion diameters for all-WM and WA-$\{2,4\}$ are approximately $3.81946$, whereas those for WA-$\{3,5\}$ and all-WA, together with the equilibrium residuals of all four configurations, are close to zero within floating-point precision. The simulated solutions illustrate why convergence to equilibrium need not imply consensus.

\subsection{Study 2: Effects of influence weights on consensus}

\begin{figure*}[t]
 \centering
 \includegraphics[width=\textwidth]{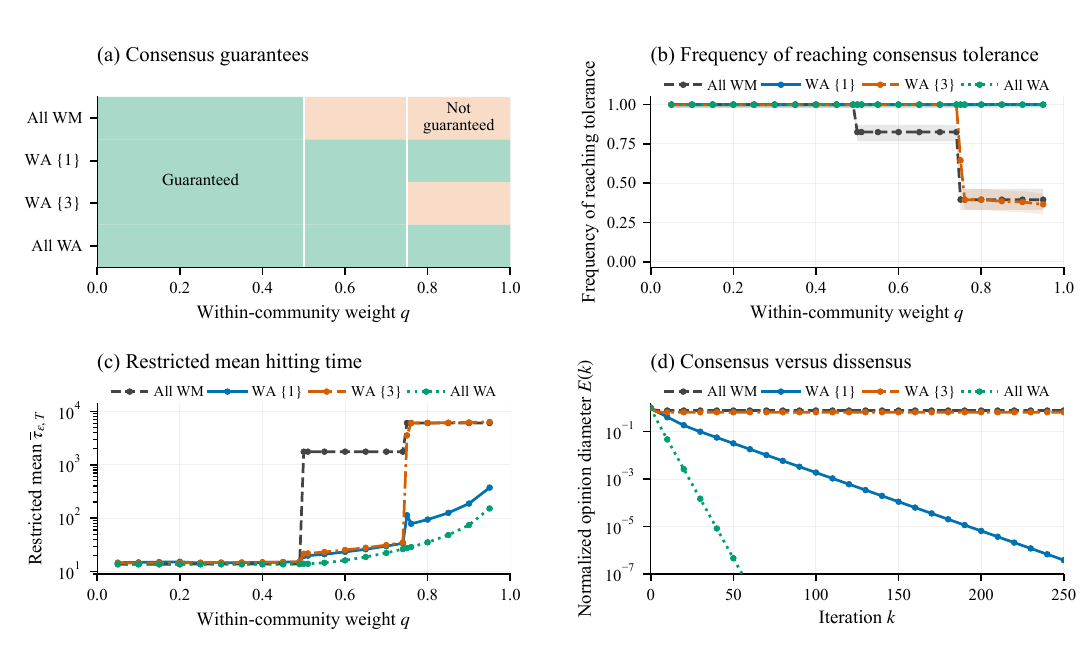}
 \caption{Study~2. Effects of influence weights on consensus in a five-node network with two communities $L=\{1,2\}$ and $U=\{3,4,5\}$. The adjacency matrix is $A(q)$ in~\eqref{eq:controlledweights}, with $0<q<1$. The four WA--WM configurations are all-WM, WA-$\{1\}$, WA-$\{3\}$, and all-WA. All use $\gamma_i=0.5$. For panels~(b,c), the same 200 independently sampled initial opinion vectors $x^{(r)}(0)\sim\mathrm{Unif}([0,1]^5)$, $r=1,\ldots,200$, are used across all 23 sampled $q$ values and four configurations, with $\varepsilon=10^{-4}$ and $T=10^4$. Panel~(d) uses $q=0.75$ and $x(0)=(0.05,0.15,0.8,0.7,0.9)^\top$, with the same inertia coefficients and simulation horizon. (a) Parameter ranges guaranteeing consensus from any initial condition according to Theorem~\ref{thm:consensus}; the boundary values $q=1/2$ for all-WM and $q=3/4$ for WA-$\{3\}$ are excluded. (b) Observed frequency of reaching the prescribed consensus tolerance $\widehat p_{\varepsilon,T}(q,\Vwa)$. (c) Restricted mean hitting time $\overline{\tau}_{\varepsilon,T}(q,\Vwa)$. (d) Normalized opinion diameter $E(k)$. Colors and line styles identify the same configurations throughout. Shading in panel~(b) indicates pointwise 95\% confidence intervals (Wilson intervals). Simulated solutions not reaching the prescribed consensus tolerance contribute $T$ to the restricted mean hitting time in panel~(c). In panels~(b,c), filled circles mark the 23 sampled $q$ values, and line segments connect the statistics at adjacent sampled values. In panel~(d), line segments connect $E(k)$ at consecutive integer time steps, and filled circles are shown every 10 steps.}\label{fig:controlled}
\end{figure*}

We next examine how influence weights affect the consensus guarantee and finite-horizon simulation outcomes for four fixed WA--WM configurations.

\paragraph{Network and consensus conditions} Partition five individuals into communities $L=\{1,2\}$ and $U=\{3,4,5\}$. For $0<q<1$, each individual assigns total weight $q$ uniformly within its community, including itself, and total weight $1-q$ uniformly to the other community:
\begin{equation}\label{eq:controlledweights}
 A(q)=\begin{pmatrix}
 \frac{q}{2}\one_2\one_2^\top&\frac{1-q}{3}\one_2\one_3^\top\\
 \frac{1-q}{2}\one_3\one_2^\top&\frac{q}{3}\one_3\one_3^\top
 \end{pmatrix},
\end{equation}
where $\one_r$ is the vector of $r$ ones. For $0<q<1$, the directed edge structure remains the same as $q$ varies; only the influence weights change. Since all influence weights are positive, a proper type-dependent cohesive set cannot contain a WA individual. Applying Theorem~\ref{thm:consensus} gives the following parameter ranges for consensus from any initial condition:
\begin{equation}\label{eq:controlledcriterion}
 \begin{array}{c|c}
 \Vwa&\text{consensus for any }x(0)\\\hline
 \varnothing&0<q<1/2\\
 \{1\}&0<q<1\\
 \{3\}&0<q<3/4\\
 \V&0<q<1
 \end{array}
\end{equation}
We establish these ranges by identifying the possible disjoint type-dependent cohesive sets.

For any $i\in\V$, the singleton set $S=\{i\}$ is not type-dependent cohesive, since $a_i(S)=a_{ii}\leq q/2<1/2$. A two-node set $S=\{i,j\}$ meeting both communities, with $i\in L$ and $j\in U$, is also not type-dependent cohesive, since $a_i(S)=q/2+(1-q)/3=(q+2)/6<1/2$. Any two disjoint type-dependent cohesive sets have at least two members each; on five nodes, at least one must have exactly two.

In the all-WM case, the only two-node subset of $L$ is $S=L=\{1,2\}$, which is type-dependent cohesive if and only if $a_i(L)=q\geq1/2$ for any $i\in L$. For any two-node subset $S=\{i,j\}\subseteq U$, we have $a_i(S)=a_j(S)=2q/3$. Thus, $S$ is type-dependent cohesive if and only if $2q/3\geq1/2$, i.e., $q\geq3/4$. For $q<1/2$, no two disjoint type-dependent cohesive sets exist, since any such pair must contain a two-node type-dependent cohesive set. For $q\geq1/2$, both $L$ and $U$ are type-dependent cohesive, with $a_i(L)=q$ for any $i\in L$ and $a_i(U)=q$ for any $i\in U$, and they are disjoint.

Under WA-$\{1\}$, any two disjoint type-dependent cohesive sets must avoid node~1. Since each contains at least two members, they would partition $\{2,3,4,5\}$ into two pairs; the pair containing node~2 would meet both communities and fail the cohesion condition. Under WA-$\{3\}$, any two such sets must similarly partition $\{1,2,4,5\}$ into two pairs. Excluding mixed-community pairs leaves only $L$ and $\{4,5\}$, which require $q\geq1/2$ and $2q/3\geq1/2$, respectively.

Finally, in the all-WA case, no nonempty proper subset is type-dependent cohesive because all entries of $A(q)$ are positive. These arguments prove~\eqref{eq:controlledcriterion}, including the excluded boundary values.

\paragraph{Simulation protocol}
We compare all-WM, WA-$\{1\}$, WA-$\{3\}$, and all-WA. The parameter $q$ is sampled from $0.05$ to $0.95$ with spacing $0.05$, together with the additional values $0.49$, $0.51$, $0.74$, and $0.76$. The same 200 independently sampled initial opinion vectors $x^{(r)}(0)\sim\mathrm{Unif}([0,1]^5)$, $r=1,\ldots,200$, are used for all 23 values of $q$ and all four WA--WM configurations. With $\gamma_i=0.5$, $\varepsilon=10^{-4}$, and $T=10^4$, this gives $23\times4\times200=18{,}400$ simulated solutions. We also simulate the four WA--WM configurations at $q=0.75$, using
$\gamma_i=0.5$, $T=10^4$, and the common initial opinion vector
$x(0)=(0.05,0.15,0.8,0.7,0.9)^\top$.

\paragraph{Simulation results and analysis}
For the network in~\eqref{eq:controlledweights}, Fig.~\ref{fig:controlled} shows that changes in the influence weights can alter the consensus guarantee and substantially affect the finite-horizon simulation outcomes.

Figure~\ref{fig:controlled}(b,c) reports the observed frequencies and restricted mean hitting times. At $q=0.50$, the all-WM configuration does not satisfy the consensus criterion in Theorem~\ref{thm:consensus}, yet 165 out of 200 simulated solutions reach the prescribed consensus tolerance. This observation does not contradict the criterion, which concerns every initial condition and asymptotic consensus rather than finitely many sampled initial conditions and reaching a prescribed consensus tolerance within a finite number of time steps. At $q=0.75$, 129 of the 200 simulated solutions under WA-$\{3\}$ reach the prescribed consensus tolerance, compared with all 200 under each of WA-$\{1\}$ and all-WA. Their restricted mean hitting times $\overline{\tau}_{\varepsilon,T}(q,\Vwa)$ defined in Section~\ref{sec:simulationprotocol} are $3565.65$, $113.655$, and $27.545$ steps. Thus, even when two configurations both guarantee consensus, their restricted mean hitting times may differ.

At $q=0.75$ (Fig.~\ref{fig:controlled}(d)), the disjoint type-dependent cohesive sets are $L$ and $U$ under all-WM, and $L$ and $\{4,5\}$ under WA-$\{3\}$. In both configurations, the two sets have initial opinion intervals $[0.05,0.15]$ and $[0.7,0.9]$, respectively. By Lemma~\ref{lem:groupinterval}, the opinions of the individuals in these two disjoint type-dependent cohesive sets remain separated by at least $0.55$ for all $k\in\mathbb N_0$. Since every solution converges to an equilibrium by Theorem~\ref{thm:convergence}, the simulated solutions under both configurations converge to dissensus equilibria. For WA-$\{1\}$ and all-WA, Theorem~\ref{thm:consensus} guarantees consensus.

\subsection{Study 3: Effects of WA proportion on consensus}
\begin{figure*}[t]
 \centering
 \includegraphics[width=\textwidth]{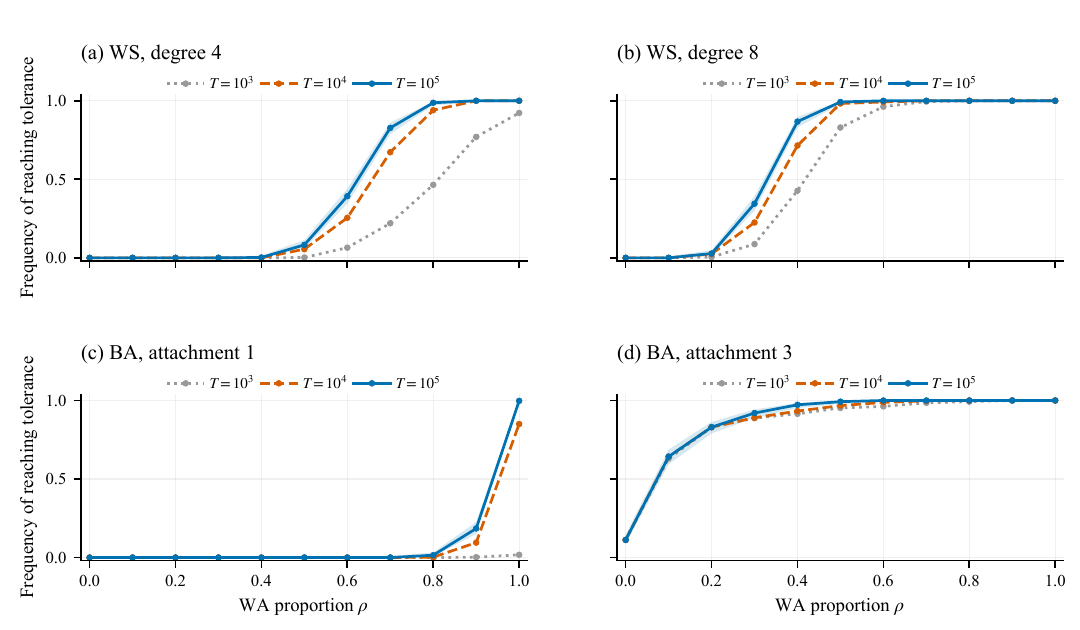}
 \caption{Study~3. Effects of WA proportion on consensus in networks of $n=100$ individuals. (a) Connected Watts--Strogatz~\cite{watts1998} networks with degree 4. (b) Connected Watts--Strogatz networks with degree 8. Both use rewiring probability 0.1. (c) Barab\'asi--Albert~\cite{barabasi1999} networks with attachment 1. (d) Barab\'asi--Albert networks with attachment 3. For each directed edge, an integer weight is independently and uniformly sampled from $\{1,\ldots,100\}$; the adjacency matrix $A$ is obtained by normalizing each row. Each panel uses 400 independent realizations, each with one initial opinion vector sampled uniformly from $[0,1]^{100}$ and a uniformly sampled node permutation. The common comparison protocol in Section~\ref{sec:simulationprotocol} is applied at the eleven WA proportions $\rho\in\{0,0.1,\ldots,1\}$, with $\gamma_i=0.5$. Curves show the observed frequency of reaching the prescribed consensus tolerance $\varepsilon=10^{-4}$ within the simulation horizons $T\in\{10^3,10^4,10^5\}$. Shading indicates pointwise 95\% confidence intervals (Wilson intervals) at $T=10^5$. Filled circles mark the eleven sampled WA proportions, and line segments connect the observed frequencies at adjacent sampled proportions.}\label{fig:fraction}
\end{figure*}
We now examine how WA proportion affects the observed frequency of reaching the prescribed consensus tolerance across four network ensembles.

\paragraph{Simulation protocol}
We consider four network ensembles with $n=100$ individuals: Watts--Strogatz (WS)~\cite{watts1998} networks with degrees 4 and 8 and rewiring probability 0.1, and Barab\'asi--Albert (BA)~\cite{barabasi1999} networks with attachment values 1 and 3. Here, WS degree denotes the initial ring degree, and BA attachment denotes the number of edges added with each new node. For each network ensemble, we generate 400 independent realizations. Each undirected edge is replaced by two directed edges in opposite directions, and self-weights are zero. For each directed edge $(i,j)$, an integer weight $b_{ij}$ is independently and uniformly sampled from $\{1,\ldots,100\}$. The row-stochastic adjacency matrix $A$ is then obtained using the integer-weight normalization procedure described in Section~\ref{sec:simulationprotocol}.
Because weights in opposite directions are sampled independently, the normalized matrix $A$ is generally asymmetric. In each realization, we sample one initial opinion vector $x(0)\sim\mathrm{Unif}([0,1]^{100})$ and select a node permutation $\pi$ uniformly at random from all permutations of $\V$. We then use the nested WA sets
\begin{equation}\label{eq:nestedtypes}
 \Vwa(\rho)=\{\pi(1),\ldots,\pi(\rho n)\},\qquad
 \rho\in\{0,0.1,\ldots,1\},
\end{equation}
where the chosen proportions make $\rho n$ an integer and $\Vwa(0)=\varnothing$. We apply the common comparison protocol in Section~\ref{sec:simulationprotocol}, with $\gamma_i=0.5$, $\varepsilon=10^{-4}$, and simulation horizons $T\in\{10^3,10^4,10^5\}$. Across the four network ensembles, this gives $4\times400=1600$ sampled networks and $4\times400\times11=17{,}600$ simulated solutions.

\paragraph{Simulation results and analysis}
Figure~\ref{fig:fraction} shows that the observed frequency is nondecreasing with $\rho$ across the eleven tested WA proportions in each of the four network ensembles at all three simulation horizons. At $\rho=0.5$ and $T=10^5$, the observed frequencies in Fig.~\ref{fig:fraction}(a) and~(b) are $8.25\%$ and $99.25\%$, respectively, whereas the frequency in Fig.~\ref{fig:fraction}(c) is only $18.5\%$ even at $\rho=0.9$. The restricted mean hitting times also differ substantially across the four network ensembles (Table~\ref{tab:waiting}). These results provide numerical evidence for the consensus-promoting role of WA individuals in these four network ensembles. Changing the WS degree or BA attachment also changes network properties, so these comparisons do not isolate the effect of degree or clustering.

\begin{table}[t]
 \centering
 \caption{Restricted mean hitting time $\overline{\tau}_{\varepsilon,T}$ defined in~\eqref{eq:restrictedmean}, with $\varepsilon=10^{-4}$ and $T=10^5$. Values average 400 independent realizations per network ensemble and are in units of $10^3$ time steps; a simulated solution not reaching the tolerance by $T$ contributes $T$ (shown as $100$).}\label{tab:waiting}
 \begin{tabular}{@{}lrrr@{}}
 \toprule
 Network ensemble&$\rho=0.5$&$\rho=0.9$&$\rho=1$\\
 \midrule
 WS, degree 4&92.755&0.914&0.593\\
 WS, degree 8&2.020&0.203&0.185\\
 BA, attachment 1&100.000&85.424&6.420\\
 BA, attachment 3&1.668&0.053&0.051\\
 \bottomrule
 \end{tabular}
\end{table}

\begin{figure*}[t]
 \centering
 \includegraphics[width=\textwidth]{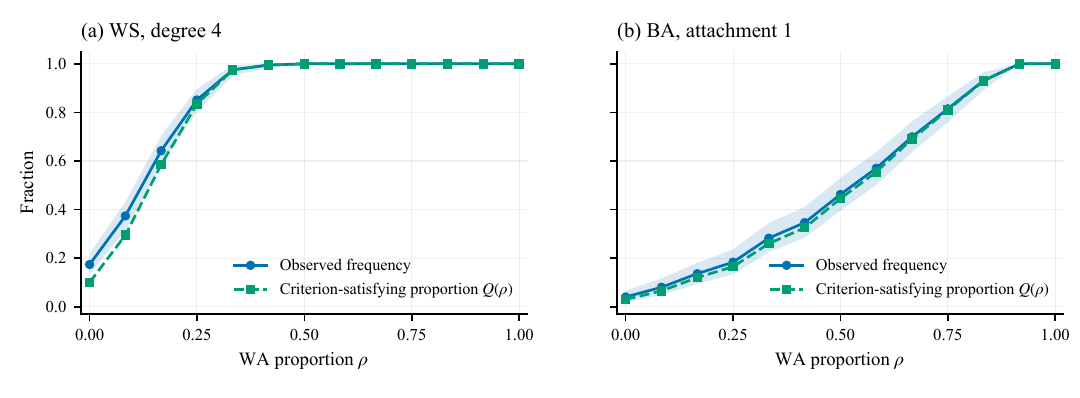}
 \caption{Study~4. Comparison of the theoretical consensus criterion and simulation results in networks of $n=12$ individuals. (a) Connected Watts--Strogatz networks with degree 4 and rewiring probability 0.1. (b) Barab\'asi--Albert networks with attachment 1. For each directed edge, an integer weight is independently and uniformly sampled from $\{1,\ldots,100\}$; the adjacency matrix $A$ is obtained by normalizing each row. Each panel uses 200 independent realizations, each with 20 initial opinion vectors independently sampled uniformly from $[0,1]^{12}$ and one uniformly sampled node permutation. The common comparison protocol in Section~\ref{sec:simulationprotocol} is applied at the 13 WA proportions $\rho\in\{0,1/12,\ldots,1\}$, with $\gamma_i=0.5$, $\varepsilon=10^{-4}$, and $T=10^5$. In both panels, green curves show the proportion $Q(\rho)$ of the 200 sampled networks whose corresponding WA--WM configurations satisfy the consensus criterion in Theorem~\ref{thm:consensus}. Blue curves show the observed frequency of reaching the prescribed consensus tolerance $\widehat p_{\varepsilon,T}(\rho)$, averaged over the 20 initial opinion vectors within each network and then over the 200 networks. The criterion and simulations use the same adjacency matrix and WA--WM configuration for each network and proportion. Shading indicates pointwise 95\% confidence intervals calculated from 5000 network-level bootstrap resamples. Filled markers indicate the 13 sampled WA proportions, and line segments connect the values at adjacent sampled proportions.}\label{fig:structure}
\end{figure*}

However, some individual realizations do not follow the overall trend: in seven cases (three in WS networks with degree 8 and four in BA networks with attachment 3), the simulated solution with a lower WA proportion $\rho$ reaches the prescribed consensus tolerance, whereas that with the higher proportion $\rho+0.1$ does not. To investigate whether this phenomenon is caused by numerical rounding errors, Supplementary Section~S2.1 examines one of these seven cases using a 50-digit computation over 2000 time steps, with the original network, initial condition, and WA--WM configurations unchanged. The higher-precision computation reproduces this difference at $\varepsilon=10^{-4}$, although the simulated solutions under both configurations reach the tolerance $\varepsilon=10^{-3}$. This nonmonotonic behavior in individual realizations does not contradict Remark~\ref{rem:conversion}: increasing the WA proportion does not necessarily guarantee consensus from a given initial opinion vector when the original WA--WM configuration does not already guarantee consensus from every initial condition. Moreover, even when the simulated solutions under both configurations converge to consensus, the solution at the higher WA proportion may take longer to reach the prescribed consensus tolerance and therefore fail to reach it within the simulation horizon.

Supplementary Section~S2 reports the numbers of simulated solutions reaching the prescribed consensus tolerance at all eleven WA proportions, together with paired comparisons for each network ensemble, while Supplementary Section~S3 reports further simulations under eight settings beyond those considered here, varying network size, inertia coefficients, or initial-opinion distributions in two network families.

\subsection{Study 4: Comparison of the theoretical consensus criterion and simulation results}
Next, we compare the graph-theoretic consensus criterion with the observed frequency of reaching the prescribed consensus tolerance using smaller networks than in Study~3, with $n=12$ individuals, so that all nonempty node subsets can be checked for type-dependent cohesion.

\paragraph{Simulation protocol}
We consider two network ensembles: WS networks with degree 4 and rewiring probability 0.1, and BA networks with attachment 1. For each network ensemble, we generate 200 independent realizations, with 20 independently sampled initial opinion vectors from $[0,1]^{12}$ per sampled network. The directed integer weights and node permutation are sampled as in Study~3. We apply the common comparison protocol in Section~\ref{sec:simulationprotocol}, with $\gamma_i=0.5$, $\varepsilon=10^{-4}$, and $T=10^5$, and 13 WA proportions $\rho\in\{0,1/12,\ldots,1\}$, corresponding to $|\Vwa|\in\{0,1,\ldots,12\}$. This gives $2\times200\times13\times20=104{,}000$ simulated solutions across 400 networks.

For each network and WA--WM configuration, we exhaustively enumerate all nonempty node subsets to evaluate the consensus criterion in Theorem~\ref{thm:consensus}, using the same adjacency matrix $A$ as in the corresponding numerical simulations. The type-dependent cohesion conditions in Definition~\ref{def:cohesion} are checked using the integer-weight comparisons described in Section~\ref{sec:simulationprotocol}. 
For each WA proportion $\rho$ and network ensemble, let $Q(\rho)$ denote the proportion of the 200 sampled networks whose corresponding WA--WM configurations satisfy the consensus criterion in Theorem~\ref{thm:consensus}.
To calculate the pointwise 95\% confidence intervals, we use the network-level bootstrap procedure in Section~\ref{sec:simulationprotocol}, with 5000 resamples of 200 realizations each.

\paragraph{Simulation results and analysis}
Figure~\ref{fig:structure} shows that the proportion $Q(\rho)$ of sampled networks whose corresponding WA--WM configurations satisfy the theoretical consensus criterion can differ from the observed frequency $\widehat p_{\varepsilon,T}(\rho)$ of reaching the prescribed consensus tolerance in numerical simulations.
For WS networks at $\rho=0$ (Fig.~\ref{fig:structure}(a)), $Q(0)=10\%$, whereas 17.4\% of simulated solutions reach the tolerance. For BA networks at $\rho=0.5$ (Fig.~\ref{fig:structure}(b)), $Q(0.5)=44.5\%$, whereas 46.225\% of simulated solutions reach the tolerance. 
All 68,720 simulated solutions under configurations satisfying the consensus criterion reach the prescribed consensus tolerance; 1512 simulated solutions under configurations failing the criterion also reach it. The latter do not contradict Theorem~\ref{thm:consensus}: the criterion concerns asymptotic consensus from any initial condition, whereas these results concern sampled initial conditions within a finite horizon. Supplementary Section~S1 gives the numbers reaching the tolerance and the total numbers simulated, grouped by the number of WA individuals and criterion status.

\section{Conclusion}\label{sec:conclusion}
For the synchronous mixed WA--WM dynamics, we characterized equilibria through graph-theoretic conditions on opinion level sets and WA balance equations and proved convergence to Lyapunov-stable equilibria. The absence of disjoint type-dependent cohesive sets provides a necessary and sufficient graph-theoretic criterion for consensus from any initial condition; a disjoint pair instead yields an open set of initial conditions leading to persistent dissensus. Converting WM individuals to WA preserves existing consensus guarantees from any initial condition. If the conversion eliminates every disjoint pair of type-dependent cohesive sets, the converted system guarantees consensus from any initial condition. These results explain the consensus-promoting role of WA individuals.

The four numerical studies examine the effects of WA locations, influence weights, and WA proportion on consensus, and compare the consensus criterion with finite-horizon simulation outcomes. In particular, Study~3 provides numerical evidence for the consensus-promoting role of WA individuals: in the four tested network ensembles, the observed frequency of reaching the prescribed consensus tolerance is nondecreasing with WA proportion at each tested simulation horizon, although individual realizations need not follow this trend. Future work includes switching between WA and WM mechanisms, time-varying networks, and convergence-time estimates.

\bibliographystyle{elsarticle-num}
\bibliography{references}

\input{supplement}

\end{document}

%% file: supplement.tex
\clearpage
\newgeometry{left=22mm,right=22mm,top=22mm,bottom=22mm,headheight=12pt,headsep=25pt,footskip=30pt}
\onecolumn
\setcounter{section}{0}
\setcounter{subsection}{0}
\setcounter{subsubsection}{0}
\setcounter{figure}{0}
\setcounter{table}{0}
\setcounter{equation}{0}
\renewcommand{\thesection}{S\arabic{section}}
\renewcommand{\thetable}{S\arabic{table}}
\renewcommand{\thefigure}{S\arabic{figure}}
\renewcommand{\theequation}{S\arabic{equation}}
\renewcommand{\theHsection}{supp.\arabic{section}}
\renewcommand{\theHsubsection}{\theHsection.\arabic{subsection}}
\renewcommand{\theHsubsubsection}{\theHsubsection.\arabic{subsubsection}}
\renewcommand{\theHfigure}{supp.\arabic{figure}}
\renewcommand{\theHtable}{supp.\arabic{table}}
\renewcommand{\theHequation}{supp.\arabic{equation}}
\setlength{\parindent}{1em}
\fontsize{11pt}{13.6pt}\selectfont

\begin{center}
{\Large Supplementary material}\\[6pt]
{\large Equilibria and Consensus in Mixed Weighted-Average\\
and Weighted-Median Opinion Dynamics:\\
The Consensus-Promoting Role of Weighted-Average Individuals}
\end{center}

This supplement provides the simulation setup, complete counts, uncertainty estimates, sensitivity analyses, and implementation checks supporting the four numerical studies in the main manuscript. The model and numerical measures are restated so that these results can be read independently; the theoretical results and proofs remain in the main manuscript.

Section~\ref{sec:sprotocol} describes the common setup, the protocols and complete counts for Studies~2 and~4, and the implementation checks. Sections~\ref{sec:paired} and~\ref{sec:sensitivity} provide complete counts, paired comparisons, and sensitivity analyses for Study~3. Section~\ref{sec:times} reports hitting-time distributions and an analytical all-WA benchmark. Section~\ref{sec:reproduction} identifies the reproduction files and states the scope of the numerical evidence.

\section{Simulation setup, complete counts, and implementation checks}\label{sec:sprotocol}
\subsection{Common simulation setup and numerical measures}\label{sec:sshared}

\paragraph{Model and notation}
The influence network is represented by a directed and weighted graph on the node set $\mathcal V=\{1,\ldots,n\}$, where node $i$ represents individual $i$. Its fixed row-stochastic adjacency matrix is $A=(a_{ij})$, with $a_{ij}\geq0$ and $\sum_{j=1}^n a_{ij}=1$ for every $i\in\mathcal V$. The influence weight $a_{ij}$ specifies how much individual $i$ takes individual $j$'s opinion into account; a directed edge from $i$ to $j$ means that $i$ listens to $j$. At time $k\in\mathbb N_0$, where $\mathbb N_0=\{0,1,2,\ldots\}$, the opinion vector is $x(k)=(x_1(k),\ldots,x_n(k))^\top\in\mathbb R^n$. The fixed partition $(\mathcal V^{\mathrm{wa}},\mathcal V^{\mathrm{wm}})$ specifies which individuals use the weighted-average (WA) and weighted-median (WM) opinion-aggregation mechanisms; this partition is the WA--WM configuration.

The base response map is $F_\zero(x)=(F_{\zero,1}(x),\ldots,F_{\zero,n}(x))^\top$. The WA response is $F_{\zero,i}(x)=\sum_j a_{ij}x_j$. For a WM individual, the weighted-median interval is
\[
 \mathcal M_i(x;A)=\operatorname*{arg\,min}_{z\in\mathbb R}
       \sum_{j=1}^n a_{ij}|z-x_j|.
\]
The WM response is the Euclidean projection~\cite[Section~8.1.1]{boyd2004} of $x_i$ onto this interval: $F_{\zero,i}(x)=x_i$ when $x_i\in\mathcal M_i(x;A)$, and otherwise $F_{\zero,i}(x)$ is the interval endpoint nearest to $x_i$. Thus an exact half-weight tie does not require choosing a predetermined lower or upper endpoint.

For a fixed inertia vector $\gamma=(\gamma_1,\ldots,\gamma_n)^\top$ satisfying $0<\gamma_i<1$, the opinion update map is
\begin{equation}\label{eq:supdate}
 x(k+1)=F_\gamma(x(k)),\qquad
 F_{\gamma,i}(x)=\gamma_i x_i+(1-\gamma_i)F_{\zero,i}(x),
 \quad i\in\mathcal V.
\end{equation}
The no-inertia endpoint $\gamma=\zero$ gives the base response map $F_\zero$. The consensus criterion and all simulations reported here use $0<\gamma_i<1$. Larger $\gamma_i$ puts more weight on the current opinion and less on the response. All individuals compute their responses from the same opinion vector $x(k)$ and update synchronously. Unless otherwise specified, $\gamma_i=0.5$.

An equilibrium is an opinion vector $x^\star$ satisfying $F_\gamma(x^\star)=x^\star$. An opinion vector whose components are all equal, $c\mathbf1$, is called a consensus vector, where $c\in\mathbb R$ and $\mathbf1$ is the vector of $n$ ones. Every consensus vector is a consensus equilibrium; an equilibrium whose components are not all equal is a dissensus equilibrium. A solution converges to consensus if $x(k)\to c\mathbf1$ for some scalar $c$.

\paragraph{Sampling units and paired comparisons}
With the adjacency matrix, WA--WM configuration, inertia vector, and initial opinion vector fixed, synchronous iteration of~\eqref{eq:supdate} produces a simulated solution. The notation WA-$S$ means that exactly the individuals indexed by $S\subseteq\mathcal V$ use WA; all others use WM. In the all-WA and all-WM configurations, every individual uses the same opinion-aggregation mechanism.

For Studies~3 and~4, a specified network model, size, and sampling procedure define a network ensemble. Each influence network generated by that procedure is a sampled network. A realization contains one sampled network and the other random inputs specified for that study. One realization may produce several simulated solutions, for example by varying the WA--WM configuration or the initial opinion vector.

Define the WA proportion as $\rho=|\mathcal V^{\mathrm{wa}}|/n$. In Studies~3 and~4, a fixed permutation of the nodes of each sampled network determines nested WA sets $\mathcal V^{\mathrm{wa}}(\rho)$: the first $\rho n$ nodes index WA individuals. Increasing $\rho$ therefore adds WA individuals while retaining those already selected. Within each comparison, the network model, network size, initial-condition distribution, and inertia coefficients are held fixed. The specific network, influence weights, and initial opinions are determined by the sampling procedure and realization. For simplicity, these settings are omitted from the notation, and the resulting statistics are written as functions of $\rho$. This does not mean that outcomes depend on $\rho$ alone, since the locations of WA individuals can also affect the results.

In Study~3 and its sensitivity analyses, a realization contains an influence network, an initial opinion vector, an inertia vector, and a node permutation. These inputs are sampled independently across realizations using the separate random streams specified below. Within a realization, all eleven proportions $\rho=0,0.1,\ldots,1$ reuse these inputs. In Study~4, a realization contains an influence network on $n=12$ nodes, a node permutation, and 20 independently sampled initial opinion vectors; the same inputs are reused at all thirteen WA counts $c=0,\ldots,12$, with $\rho=c/12$. Each simulated solution restarts from its corresponding original initial opinion vector, rather than continuing from the endpoint of another solution. Table~\ref{tab:scope} distinguishes the number of sampled networks from the number of simulated solutions.

\begin{table}[htbp]
\centering
\caption{Summary of simulations and computational checks. Comparisons across tolerances and horizons reuse recorded hitting times rather than generating new simulated solutions.}\label{tab:scope}
\small
\begin{tabular}{@{}>{\raggedright\arraybackslash}p{31mm}>{\raggedright\arraybackslash}p{83mm}>{\raggedright\arraybackslash}p{42mm}@{}}
\toprule
Study & Design & Evaluations\\
\midrule
Study 1: WA locations & Six individuals; four WA--WM configurations; 2000 updates each & 4 simulated solutions\\
Study 2: influence weights & Five individuals; 23 values of $q$; four configurations; 200 initial vectors & 18,400 simulated solutions\\
Study 2: separated opinions & Five individuals, $q=0.75$; four WA--WM configurations & 4 simulated solutions\\
Study 3: WA proportion & Four ensembles; 400 realizations each; 11 proportions & 17,600 simulated solutions\\
Study 4: subset enumeration & Two ensembles; 200 realizations each; 13 WA counts & 5200 network/configuration checks\\
Study 4: initial conditions & 400 networks; 20 initial vectors per network; 13 counts & 104,000 simulated solutions\\
Sensitivity (Study 3) & Eight settings; 200 realizations each; 11 proportions & 17,600 simulated solutions\\
Precision check (Study 3) & One saved realization; two configurations; 2000 updates in double and 50-digit arithmetic & 2 configuration replays\\
All-WA benchmark (Study 1) & Six individuals, all WA; 11 coefficient vectors & 11 simulated solutions\\
\bottomrule
\end{tabular}
\end{table}

\paragraph{Opinion diameter and equilibrium residual}
For any opinion vector $x$, define the opinion diameter by $D(x)=\max_i x_i-\min_i x_i$. When $D(x(0))>0$, the normalized opinion diameter along a simulated solution is
\begin{equation}\label{eq:sdiameter}
 E(k)=\frac{D(x(k))}{D(x(0))}.
\end{equation}
For a tolerance $\varepsilon>0$, its first hitting time is
\begin{equation}\label{eq:shitting}
 \tau_\varepsilon=\inf\{k\in\mathbb N_0:E(k)\leq\varepsilon\},
\end{equation}
with $\tau_\varepsilon=\infty$ if the tolerance is never reached. For an initial consensus vector, we set $E(k)=0$ for all $k\in\mathbb N_0$ and $\tau_\varepsilon=0$.

The equilibrium residual is
\begin{equation}\label{eq:sresidual}
 R(x)=\max_i|F_{\zero,i}(x)-x_i|.
\end{equation}
Let $\gamma_{\min}=\min_i\gamma_i$ and $\gamma_{\max}=\max_i\gamma_i$. The update residual satisfies
\[
 (1-\gamma_{\max})R(x)\leq\|F_\gamma(x)-x\|_\infty
 \leq(1-\gamma_{\min})R(x).
\]
Since each $1-\gamma_i>0$, $R(x)=0$ is equivalent to the equilibrium equation $F_\gamma(x)=x$. The equilibrium residual measures violation of the equilibrium equation, whereas the opinion diameter measures the spread of current opinions. A dissensus equilibrium has zero residual and positive opinion diameter. In the precision diagnostic, the normalized equilibrium residual means $R(x(k))/D(x(0))$, using the same initial opinion diameter as~\eqref{eq:sdiameter}; this quantity is set to zero for an initial consensus vector. Neither the equilibrium residual nor the opinion update residual is a distance to the limiting opinion vector.

\paragraph{Stopping and finite-horizon classification}
Studies~2--4 use the default consensus tolerance $\varepsilon=10^{-4}$. The simulation horizon is $T=10^4$ updates in Study~2 and $T=10^5$ updates in Studies~3 and~4. During each simulated solution, the first hitting times for all three tolerances $10^{-3}$, $10^{-4}$, and $10^{-5}$ are recorded. A simulation may stop when it reaches the strictest tolerance, when the floating-point opinion vector no longer changes, or when it reaches the horizon. The precision check and all-WA benchmark use the update counts specified in their respective sections.

For each reported tolerance and horizon, the observed frequency is the number of simulated solutions with $\tau_\varepsilon\leq T$ divided by the total number of simulated solutions in that comparison. Solutions that have not reached the tolerance remain in the denominator. This finite-horizon classification does not establish persistent dissensus, and floating-point stationarity or a small equilibrium residual does not certify an exact equilibrium.

\subsection{Network generation and implementation}\label{sec:snetwork}

\paragraph{Network generation}
Networks are generated with NetworkX 3.2.1~\cite{networkx321,hagberg2008}. Watts--Strogatz (WS) networks~\cite{watts1998} use \path{watts_strogatz_graph}, rewiring probability 0.1, and degree 4 or 8 (the degree of each node in the initial ring). Disconnected realizations are discarded and resampled; the number of attempts is recorded. Barab\'asi--Albert (BA) networks~\cite{barabasi1999} use \path{barabasi_albert_graph}, attachment $m_{\mathrm{BA}}=1$ or 3, and the default initial star on $m_{\mathrm{BA}}+1$ nodes. Each undirected edge is replaced by two directed edges, and self-loops are excluded. The two directions receive independent integer weights, uniformly distributed on $\{1,\ldots,100\}$, which are then normalized by row. Write $B=(b_{ij})$ for these unnormalized weights and $A=(a_{ij})$ for the adjacency matrix, where $a_{ij}=b_{ij}/\sum_h b_{ih}$. The resulting matrices need not be symmetric, but their edge sets are bidirectional; the theory also allows other directed networks. Initial opinions in the baseline settings of Studies~3 and~4 are independent $\mathrm{Unif}([0,1])$ draws, and permutations are sampled uniformly from all node orderings. The deterministic network, initial vector, and heterogeneous inertia vector of Study~1 are also specified in Section~\ref{sec:walimit} below.

\paragraph{Random streams}
The original simulation batch contains the Study~3 realizations, the $n=12$ networks and permutations for Study~4, and the Study~3 sensitivity realizations. Its master seed is 20260926, with separate streams for graphs, weights, initial opinions, inertia coefficients, and node permutations. These streams use NumPy's \texttt{default\_rng} with \texttt{SeedSequence} entropy
\[
 (20260926,f,n,d,r,s).
\]
Here $f=1$ for WS and $f=2$ for BA, $d$ is the WS degree or BA attachment, $r$ is the zero-based realization index, and $s=0,1,2,3,4$ selects the graph, weights, initial opinions, coefficients, or permutation stream, respectively. The graph stream draws successive integer seeds from $\{0,\ldots,2^{32}-2\}$ until a connected graph is obtained. The weight stream draws an $n\times n$ integer array on $\{1,\ldots,100\}$ and sets entries to zero where the graph has no edge. The initial stream supplies $n$ independent uniform draws; the permutation stream supplies a uniform node ordering. In the heterogeneous-coefficient setting, the coefficient stream generates independent $v_i\sim\mathrm{Unif}([0.2,0.8])$, and the inertia coefficients are $\gamma_i=1-v_i$. The current archive stores $\gamma_i$ in its \texttt{gamma} arrays and retains the original response weights $v_i$ in \texttt{legacy\_adjustment\_gamma}. The relation $\gamma_i=1-v_i$ gives the same dynamics under the inertia convention in~\eqref{eq:supdate}; the recorded simulated solutions and numerical summaries are retained. All sampled inputs and seeds are retained for reproduction.

\paragraph{Median ties in the implementation}
The implementation evaluates the WA and WM responses defined above and then applies~\eqref{eq:supdate}. Opinions are computed in double precision unless a precision check is explicitly identified. For a set $S$ encountered while accumulating neighbor weights in increasing opinion order, the WM response is determined by comparing $2\sum_{j\in S}b_{ij}$ with $\sum_jb_{ij}$. Equality identifies an exact half-weight tie before floating-point normalization; the response is then obtained by projection onto the resulting weighted-median interval. Study~1 uses integer row sum 2; Study~2 uses the integer representation given below; and Studies~3 and~4, including the sensitivity analyses, use sampled integer weights. These exact comparisons concern the weights, not the arithmetic used for opinions.

\paragraph{Structural criterion}
For a nonempty set $S\subseteq\mathcal V$, its internal weight at individual $i$ is $a_i(S)=\sum_{j\in S}a_{ij}$. The set is type-dependent cohesive if $a_i(S)\geq1/2$ for every $i\in S\cap\mathcal V^{\mathrm{wm}}$ and $a_i(S)=1$ for every $i\in S\cap\mathcal V^{\mathrm{wa}}$. The consensus criterion in the main manuscript states that the system converges to consensus from any initial opinion vector if and only if there are no two disjoint nonempty type-dependent cohesive sets. This is the criterion evaluated by exhaustive subset enumeration in Study~4. For integer weights, the subset routine compares a WM member's internal weight with the ceiling of half its row sum, and a WA member's internal weight with the full row sum. The former comparison avoids overflow from doubling and implements the same half-weight condition exactly.

\subsection{Protocols and complete counts for Studies 2 and 4}
The next two protocols specify the weight variation in Study~2 and the multiple-initial-condition comparison in Study~4.

\paragraph{Study 2: influence weights}
The five individuals form the communities $L=\{1,2\}$ and $U=\{3,4,5\}$. Each individual assigns total weight $q$ uniformly to its own community, including itself, and total weight $1-q$ uniformly to the other community. Thus $q$ is a within-community influence parameter, not a proportion of WA individuals. For $q=t/100$, the adjacency matrix is represented exactly as $A=B/600$, where
\[
 b_{ij}=\begin{cases}
 3t,&i,j\in L,\\
 2(100-t),&i\in L,\ j\in U,\\
 3(100-t),&i\in U,\ j\in L,\\
 2t,&i,j\in U.
 \end{cases}
\]
Every row of $B$ sums to 600. These integer weights specify the rational influence weights before floating-point evaluation and permit exact comparisons at the median threshold. The 23 simulation numerators $t$ are
\[
5,10,15,20,25,30,35,40,45,49,50,51,55,60,65,70,74,75,76,80,85,90,95.
\]
The four configurations are all-WM, WA-$\{1\}$, WA-$\{3\}$, and all-WA. All four configurations and all 23 values of $q$ use the same 200 independently sampled initial opinion vectors $x^{(r)}(0)\sim\mathrm{Unif}([0,1]^5)$, $r=1,\ldots,200$. They are generated by \texttt{default\_rng(SeedSequence([20261002,20]))} as a $200\times5$ array of independent draws. Each simulated solution restarts from its corresponding initial opinion vector, with $\gamma_i=0.5$ and horizon $T=10^4$. At fixed $q$, only the WA set changes across the four configurations; across values of $q$, the influence weights change as specified above. 

WA-$\{1\}$ and WA-$\{3\}$ have the same WA proportion but different WA locations. We therefore write the hitting time, observed frequency, and restricted mean hitting time as $\tau_\varepsilon^{(r)}(q,\mathcal V^{\mathrm{wa}})$, $\widehat p_{\varepsilon,T}(q,\mathcal V^{\mathrm{wa}})$, and $\overline{\tau}_{\varepsilon,T}(q,\mathcal V^{\mathrm{wa}})$, respectively, with $\varepsilon=10^{-4}$ and $N=200$ independent initial opinion vectors at each parameter setting.

The illustrative simulated solutions at $q=0.75$ with initially separated opinions use the initial opinion vector
\[
 x(0)=(0.05,0.15,0.8,0.7,0.9)^\top.
\]
These four deterministic simulated solutions are separate from the 200-initial-vector estimates.

\paragraph{Study 4: multiple initial conditions on each network}
This study uses the two $n=12$ ensembles from the original simulation batch: connected WS networks with degree 4 and BA networks with attachment 1. The networks, directed weights, and node permutations remain unchanged. For each network, 20 independent initial opinion vectors are sampled from $\mathrm{Unif}([0,1]^{12})$ and reused at all thirteen WA counts. For family index $f'=0$ (WS) or $1$ (BA) and zero-based network index $r$, the new initial array is generated by \texttt{default\_rng} from a \texttt{SeedSequence} with entropy $(20261002,10,f',r)$. The same 200 networks in each ensemble are evaluated at each count, giving $2\times200\times13\times20=104{,}000$ simulated solutions.

For the following estimate, index the networks by $r=1,\ldots,200$ and their initial vectors by $s=1,\ldots,20$. At WA proportion $\rho=c/12$, let $I_{rs}(\rho)$ equal one if the simulated solution from network $r$ and initial opinion vector $s$ reaches $E(k)\leq\varepsilon$ by $T$, and zero otherwise, with $\varepsilon=10^{-4}$ and $T=10^5$. The plotted estimate is
\begin{equation}
 \widehat p_{\varepsilon,T}(\rho)=\frac1{200}\sum_{r=1}^{200}
 \left(\frac1{20}\sum_{s=1}^{20}I_{rs}(\rho)\right).
\end{equation}
The 5000 bootstrap~\cite{efron1979} resamples select 200 whole network realizations with replacement, retaining all associated initial vectors and WA counts. The bootstrap generators use seeds 20261002 for WS and 20261003 for BA. Resampling whole realizations preserves the dependence among simulated solutions on each network and the pairing across WA counts. The bands are pointwise 95\% percentile intervals for this sampling design; they are not simultaneous bands over all counts. An interval reduced to one does not establish that the underlying probability is one. Table~\ref{tab:conditional} reports pooled counts, with the network dependence stated explicitly.

\begin{table}[htbp]
\centering
\caption{Numbers of simulated solutions reaching normalized opinion diameter $10^{-4}$ by $10^5$ updates, pooled over all 13 WA counts. The criterion is the absence of two disjoint nonempty type-dependent cohesive sets. Each denominator is the total number of simulated solutions in the indicated category. Observations sharing a network are dependent.}\label{tab:conditional}
\begin{tabular}{@{}lrr@{}}
\toprule
Ensemble & Criterion satisfied & Criterion not satisfied\\
\midrule
WS, degree 4 & 43,140 / 43,140 & 906 / 8860\\
BA, attachment 1 & 25,580 / 25,580 & 606 / 26,420\\
\bottomrule
\end{tabular}
\end{table}

Table~\ref{tab:controlledcounts} gives the simulated solution counts for all 23 values of $q$ in Study~2. Table~\ref{tab:structurecounts} gives the structural and simulated solution counts for Study~4. At each WA count $c$ in one ensemble, each of the 200 sampled networks has its corresponding WA--WM configuration. Let $G$ be the number of sampled networks whose corresponding configurations contain no two disjoint nonempty type-dependent cohesive sets. Then $Q(\rho)=G/200$, with $\rho=c/12$, is the proportion of sampled networks for which the criterion guarantees consensus from any initial opinion vector. In contrast, $\widehat p_{\varepsilon,T}(\rho)$ is the observed frequency of reaching the prescribed consensus tolerance within the simulation horizon for the sampled initial opinion vectors, with denominator $200\times20=4000$ simulated solutions. The simulated solution denominators within the two criterion categories are $20G$ and $20(200-G)$. The network realization remains the independent sampling unit; the 20 simulated solutions on each network are not treated as 20 independent network draws.

\begin{table}[p]
\centering
\caption{Study 2: numbers reaching $E(k)\leq10^{-4}$ by $10^4$ updates. Every entry has denominator 200; all columns and rows reuse the same initial vectors.}\label{tab:controlledcounts}
\begin{tabular}{@{}rrrrr@{}}
\toprule
$q$ & All-WM & WA-$\{1\}$ & WA-$\{3\}$ & All-WA\\
\midrule
0.05 & 200 & 200 & 200 & 200\\
0.10 & 200 & 200 & 200 & 200\\
0.15 & 200 & 200 & 200 & 200\\
0.20 & 200 & 200 & 200 & 200\\
0.25 & 200 & 200 & 200 & 200\\
0.30 & 200 & 200 & 200 & 200\\
0.35 & 200 & 200 & 200 & 200\\
0.40 & 200 & 200 & 200 & 200\\
0.45 & 200 & 200 & 200 & 200\\
0.49 & 200 & 200 & 200 & 200\\
0.50 & 165 & 200 & 200 & 200\\
0.51 & 165 & 200 & 200 & 200\\
0.55 & 165 & 200 & 200 & 200\\
0.60 & 165 & 200 & 200 & 200\\
0.65 & 165 & 200 & 200 & 200\\
0.70 & 165 & 200 & 200 & 200\\
0.74 & 165 & 200 & 200 & 200\\
0.75 & 79 & 200 & 129 & 200\\
0.76 & 79 & 200 & 79 & 200\\
0.80 & 79 & 200 & 79 & 200\\
0.85 & 79 & 200 & 77 & 200\\
0.90 & 79 & 200 & 76 & 200\\
0.95 & 79 & 200 & 73 & 200\\
\bottomrule
\end{tabular}
\end{table}

\begin{table}[p]
\centering
\caption{Study 4 at each WA count $c$, with $\rho=c/12$. For each ensemble, $G$ is the number of sampled networks whose corresponding WA--WM configurations satisfy the consensus criterion (out of 200); $S$ is the number of simulated solutions reaching the prescribed consensus tolerance within the simulation horizon (out of 4000); $S_{\rm no}$ counts those simulated solutions under configurations not satisfying the criterion, with denominator $20(200-G)$. All $20G$ simulated solutions under configurations satisfying the criterion reach the prescribed consensus tolerance within the simulation horizon. Parameters are $\varepsilon=10^{-4}$ and $T=10^5$.}\label{tab:structurecounts}
\begin{tabular}{@{}rrrrrrr@{}}
\toprule
&\multicolumn{3}{c}{WS, degree 4}&\multicolumn{3}{c}{BA, attachment 1}\\
$c$ & $G$ & $S$ & $S_{\rm no}$ & $G$ & $S$ & $S_{\rm no}$\\
\midrule
0 & 20 & 696 & 296 & 6 & 161 & 41\\
1 & 59 & 1496 & 316 & 13 & 321 & 61\\
2 & 117 & 2568 & 228 & 24 & 544 & 64\\
3 & 167 & 3405 & 65 & 33 & 734 & 74\\
4 & 195 & 3901 & 1 & 52 & 1130 & 90\\
5 & 199 & 3980 & 0 & 65 & 1385 & 85\\
6 & 200 & 4000 & 0 & 89 & 1849 & 69\\
7 & 200 & 4000 & 0 & 111 & 2281 & 61\\
8 & 200 & 4000 & 0 & 138 & 2797 & 37\\
9 & 200 & 4000 & 0 & 162 & 3259 & 19\\
10 & 200 & 4000 & 0 & 186 & 3725 & 5\\
11 & 200 & 4000 & 0 & 200 & 4000 & 0\\
12 & 200 & 4000 & 0 & 200 & 4000 & 0\\
\bottomrule
\end{tabular}
\end{table}

\subsection{Implementation checks}
\paragraph{Checks for the original simulations}
The initial implementation audit passed 25 unit tests, reproduced 396 saved simulated solutions, recomputed all 5200 network/configuration checks, and checked 800 pairs of parameter-sensitivity and baseline realizations with the same network size. All 3600 random-sample files in this batch were checked against the specified protocol and recorded hashes. Pilot samples were excluded from the reported statistics. All computed opinions remained between the smallest and largest initial opinions, consistent with the positive invariance established in the main manuscript. These checks concern the implementation and numerical records; they do not replace the mathematical proofs.

\paragraph{Checks for the additional simulation batch}
The additional simulation batch contains the controlled-weight simulated solutions in Study~2, the 20 initial vectors per archived small network in Study~4, and the all-WA benchmark described in Section~\ref{sec:walimit}. Its random sampling uses seed 20261002 and leaves all original samples unchanged. Its 400 small-network files include hashes identifying the original network and permutation inputs. The consensus criterion was recomputed for each configuration, and 260 simulated solutions were reproduced exactly: every twentieth network, all 13 WA counts, and initial-vector index 7. A further 184 simulated solutions were reproduced for the controlled-weight study, covering all 23 values of $q$ and four configurations at initial-vector indices 7 and 193. An independent exhaustive subset test checked the four specified WA--WM configurations (all-WM, WA-$\{1\}$, WA-$\{3\}$, all-WA) at each $q=t/100$, $t=1,\ldots,99$, giving 396 configurations. This additional grid is distinct from the 23-value simulation grid and does not enumerate all possible WA--WM configurations. Another 240 one-step updates near these thresholds were compared with an independent Python implementation. All opinions remained between the smallest and largest initial opinions in the multiple-initial-condition study. These deterministic checks do not add independent statistical trials.

\section{Study 3: complete counts and paired comparisons}\label{sec:paired}
For the four network ensembles in Study~3, the bands at each WA proportion are pointwise 95\% Wilson intervals~\cite{wilson1927} over independent realizations. Comparisons between proportions retain the pairing of inputs. Each ensemble contains $N=400$ independent realizations. Let $I_r(\rho)$ equal one if the simulated solution from realization $r$ reaches $E(k)\leq10^{-4}$ by $10^5$ updates, and zero otherwise. The estimated difference between proportions $\rho_a$ and $\rho_b$ is
\begin{equation}
 \widehat\Delta(\rho_a,\rho_b)
 =\frac{1}{N}\sum_{r=1}^{N}\bigl(I_r(\rho_b)-I_r(\rho_a)\bigr).
\end{equation}
We use 5000 bootstrap~\cite{efron1979} resamples of complete realizations to compute percentile intervals for these differences. The integer seeds are 20260930, 20260931, 20260932, and 20260933 for WS degree 4, WS degree 8, BA attachment 1, and BA attachment 3, respectively; each resample retains all eleven WA proportions. Table~\ref{tab:contrasts} reports two comparisons; differences between every pair of adjacent proportions are also retained in the statistical summary. The intervals are pointwise, not simultaneous. If all observed paired differences are zero, the resulting interval is reduced to zero; this does not establish that the underlying probability difference is exactly zero.

\begin{table}[htbp]
\centering
\caption{Paired differences in the frequency of reaching $E(k)\leq10^{-4}$ by $10^5$ updates, in percentage points; $N=400$ per ensemble. Parentheses contain pointwise 95\% bootstrap percentile intervals.}\label{tab:contrasts}
\begin{tabular}{@{}lrr@{}}
\toprule
Ensemble & $\rho:0\to0.5$ & $\rho:0.5\to1$\\
\midrule
WS, degree 4 & $8.25\ (5.50,11.00)$ & $91.75\ (89.00,94.50)$\\
WS, degree 8 & $99.25\ (98.25,100.00)$ & $0.75\ (0.00,1.75)$\\
BA, attachment 1 & $0.00\ (0.00,0.00)$ & $99.75\ (99.25,100.00)$\\
BA, attachment 3 & $88.00\ (84.75,91.00)$ & $0.75\ (0.00,1.75)$\\
\bottomrule
\end{tabular}
\end{table}

Tables~\ref{tab:maincounts} and~\ref{tab:robustcounts} report the counts at every WA proportion. At $\varepsilon=10^{-4}$ and $T=10^5$, the four main-study frequency curves and eight sensitivity-study frequency curves are all nondecreasing with $\rho$, providing numerical support for the consensus-promoting role of WA individuals in these sampled network ensembles. Nevertheless, in seven paired comparisons in the main data, only the WA--WM configuration with the lower of two adjacent WA proportions reaches the tolerance: two at $0.2\to0.3$ and one at $0.3\to0.4$ in WS degree-8 networks, and four at $0\to0.1$ in BA attachment-3 networks. One further case occurs at $0.6\to0.7$ for WS degree-4 networks with $n=50$. Thus the aggregate trend does not hold for every individual realization.

\begin{table}[htbp]
\centering
\caption{Numbers of main-study simulated solutions reaching $E(k)\leq10^{-4}$ by $10^5$ updates. Every entry has denominator 400.}\label{tab:maincounts}
\small
\setlength{\tabcolsep}{5pt}
\begin{tabular}{@{}lrrrrrrrrrrr@{}}
\toprule
$\rho$ & 0 & .1 & .2 & .3 & .4 & .5 & .6 & .7 & .8 & .9 & 1\\
\midrule
WS, degree 4 &0&0&0&0&1&33&157&331&395&400&400\\
WS, degree 8 &0&0&11&138&347&397&400&400&400&400&400\\
BA, attachment 1 &0&0&0&0&0&0&0&0&6&74&399\\
BA, attachment 3 &45&257&332&368&389&397&400&400&400&400&400\\
\bottomrule
\end{tabular}
\end{table}

\subsection{High-precision check of one realization}\label{sec:reversal}
After the main simulations, we select one of these seven comparisons for a precision check without generating a new sample. We first order their saved-file relative paths lexicographically and, for comparisons from the same file, order the lower WA proportions numerically. The first is \mbox{\path{main/ba_m3/sample_0211.npz}}, with $n=100$, attachment 3, uniform initial opinions, and $\gamma_i=0.5$. The saved network, directed integer weights, initial vector, and node permutation are unchanged. We compare the all-WM configuration with the configuration obtained by converting the first ten individuals indexed by the node permutation to WA. Their indices are $\{12,17,37,45,47,56,58,71,83,84\}$. This is a precision check of an observed case, not an independent trial or an estimate of how frequently such cases occur.

\begin{figure}[htp]
\centering
\includegraphics[width=\textwidth]{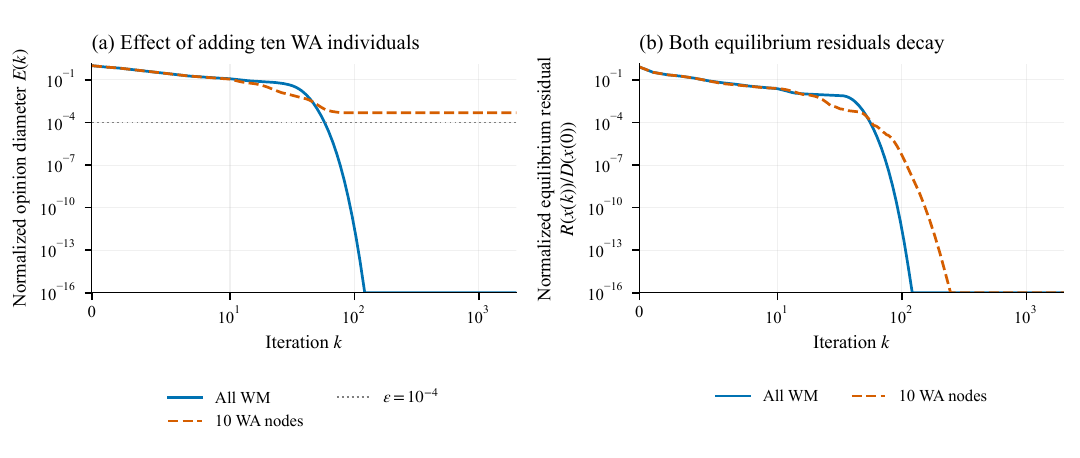}
\caption{High-precision recomputation of one saved BA realization with attachment 3, $n=100$, and $\gamma_i=0.5$. Both WA--WM configurations use the same influence network and initial opinion vector. Blue solid curves show the all-WM configuration; orange dashed curves show the configuration with ten WA individuals. (a) Normalized opinion diameter $E(k)$; the dotted line marks $10^{-4}$. (b) Normalized equilibrium residual $R(x(k))/D(x(0))$. Both panels use 50-digit arithmetic for 2000 updates, a linear time axis through step 10 and a logarithmic time axis thereafter, and a plotting floor of $10^{-16}$. Section~\ref{sec:reversal} specifies the selected realization and WA set. Each curve represents one simulated solution.}\label{fig:reversal}

\end{figure}

We replay both configurations for 2000 updates with the independent Python implementation and 50-digit decimal arithmetic. The saved binary initial opinions are converted exactly before the calculation; they are not replaced by rounded decimal inputs or newly sampled values. Median-threshold comparisons use the saved integer weights.

For the all-WM configuration, the first hitting times for $\varepsilon=(10^{-3},10^{-4},10^{-5})$ are $(50,58,65)$. With ten WA individuals, the first tolerance is reached at step 56, while the other two are not reached during the 2000-update replay. At step 2000, the configuration with ten WA individuals has normalized opinion diameter approximately $4.85618\times10^{-4}$ and normalized equilibrium residual below $1.1\times10^{-50}$. Over these 2000 updates, the largest componentwise difference between the double-precision and decimal simulated solutions is below $4.5\times10^{-16}$, after converting decimal opinion vectors to double precision for comparison.

The archived double-precision runs use horizon $10^5$ and the stopping rules in Section~\ref{sec:sshared}; those records likewise classify the all-WM configuration as reaching all three tolerances and the ten-WA configuration as reaching only $10^{-3}$. The 50-digit calculation covers 2000 updates for these two configurations only. It does not extend the high-precision check to the full original horizon or to the other six main-study cases.

Figure~\ref{fig:reversal} shows that the difference in outcomes at tolerance $10^{-4}$ persists under higher-precision computation; the remaining opinion diameter is much larger than the observed roundoff discrepancy. Both configurations reach $10^{-3}$, so changing the tolerance changes the observed outcome in this case. In general, finite-horizon comparisons can also depend on the simulation horizon and the speed at which solutions approach the tolerance. The check does not establish whether the configuration with ten WA individuals eventually reaches $10^{-4}$ or exact consensus, and its small equilibrium residual does not certify an exact dissensus equilibrium.

This example is consistent with the theoretical result on converting WM individuals to WA: the conversion preserves an existing guarantee of asymptotic consensus from any initial opinion vector, but does not guarantee that each simulated solution reaches a prescribed tolerance within a fixed horizon. All seven cases remain in the reported counts. Section~\ref{sec:sensitivity} examines the frequency trend under other simulation settings.

\section{Study 3: sensitivity to simulation parameters}\label{sec:sensitivity}
We examine whether the WA-proportion comparison depends on the simulation setting or the numerical tolerance. The first comparison changes one model input at a time; the second compares the recorded hitting times at different tolerances.

\subsection{Network size, inertia coefficients, and initial opinions}
\begin{figure}[tp]
\centering
\includegraphics[width=\textwidth,height=0.84\textheight,keepaspectratio]{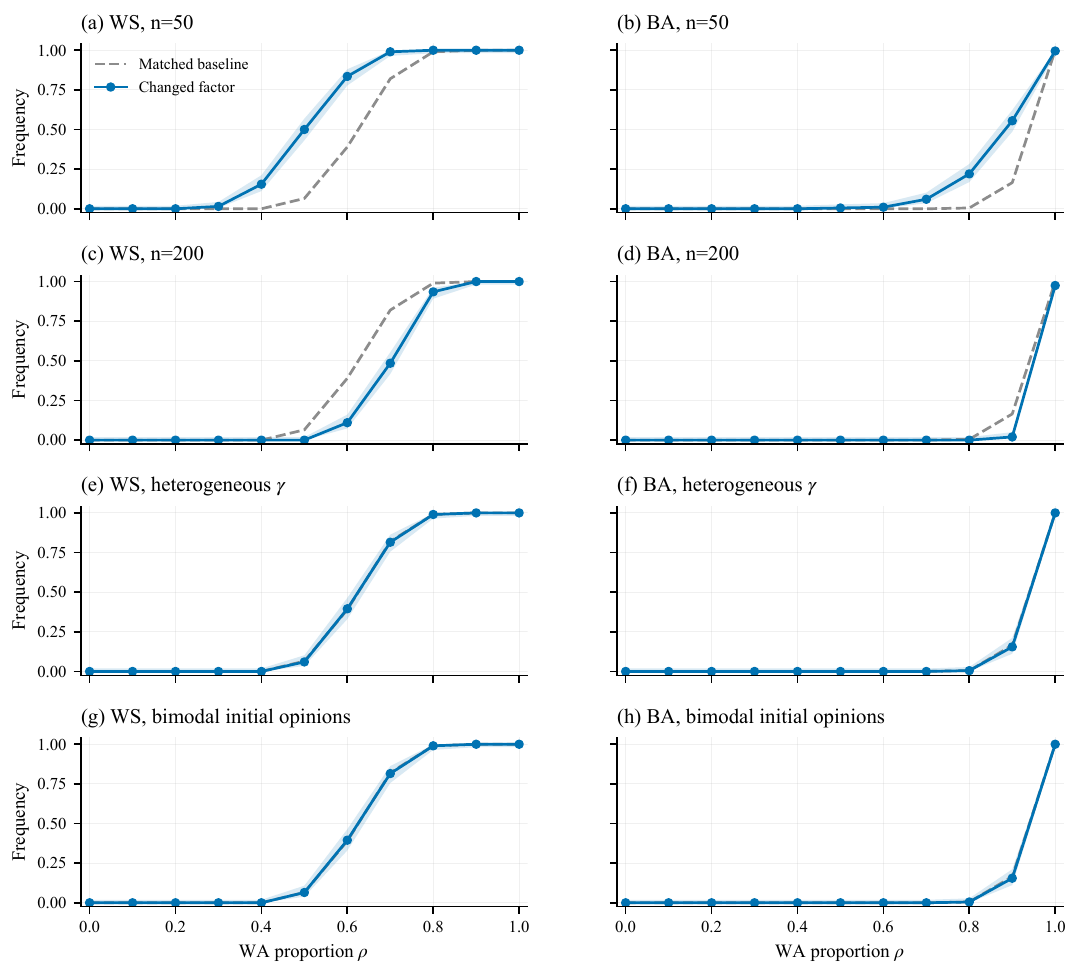}
\caption{Sensitivity of the observed frequency of reaching $E(k)\leq10^{-4}$ by $T=10^5$. The left column shows WS networks with degree 4 and rewiring probability 0.1; the right column shows BA networks with attachment 1. From top to bottom, the rows change one baseline factor: network size to $n=50$ or $n=200$, inertia coefficients to independent $\mathrm{Unif}([0.2,0.8])$ draws, or initial opinions to independent draws from the equal mixture of $\mathrm{Unif}([0,0.2])$ and $\mathrm{Unif}([0.8,1])$. Colored curves use 200 realizations in each modified setting; gray dashed curves use the first 200 baseline realizations with $n=100$, $\gamma_i=0.5$, and independent $\mathrm{Unif}([0,1])$ initial opinions. Circles mark the eleven sampled WA proportions, and line segments connect adjacent sampled proportions. Shading gives pointwise 95\% Wilson intervals. Unchanged inputs are shared with the baseline when the network size is fixed. Network generation and influence weights follow Section~\ref{sec:snetwork}.}\label{fig:robustness}

\end{figure}

The supplementary simulations use WS networks with degree 4 or BA networks with attachment 1. We vary one baseline factor at a time: network size is changed to $n=50$ or $200$; inertia coefficients are sampled independently from $\mathrm{Unif}([0.2,0.8])$; or initial opinions are sampled independently from $\tfrac12\mathrm{Unif}([0,0.2])+\tfrac12\mathrm{Unif}([0.8,1])$. The mixture assigns equal probability to the two intervals but does not require exactly half of the individuals in each. The two network families and four one-factor changes give eight settings, each with 200 independent realizations and eleven nested WA proportions. Coefficients are fixed along each simulated solution even when they differ between individuals. For bimodal opinions, a uniform draw $u$ from the original initial-opinion stream is mapped to $0.4u$ for $u<1/2$ and $0.8+0.4(u-1/2)$ otherwise. This gives the stated mixture while retaining the coupling to the baseline draw.

At fixed size $n=100$, each realization shares its unchanged inputs with the corresponding one among the first 200 baseline realizations. In particular, changing the coefficients or initial opinions leaves the network, weights, and node permutation unchanged. Networks of different sizes are generated separately. The gray reference curves in Fig.~\ref{fig:robustness} use the first 200 baseline realizations, rather than the 400 realizations in the main figure, so both curves in each comparison have the same denominator.

\begin{table}[!htbp]
\centering
\caption{Numbers of sensitivity-study simulated solutions reaching $E(k)\leq10^{-4}$ by $10^5$ updates. Every entry has denominator 200. WS networks have degree 4 and BA networks have attachment 1. Unless varied, $n=100$, $\gamma_i=0.5$, and initial opinions are independent $\mathrm{Unif}([0,1])$ draws. ``Heterogeneous'' uses independent $\gamma_i\sim\mathrm{Unif}([0.2,0.8])$; ``bimodal'' uses independent initial opinions from $\tfrac12\mathrm{Unif}([0,0.2])+\tfrac12\mathrm{Unif}([0.8,1])$.}\label{tab:robustcounts}
\small
\setlength{\tabcolsep}{4pt}
\begin{tabular}{@{}lrrrrrrrrrrr@{}}
\toprule
$\rho$ & 0 & .1 & .2 & .3 & .4 & .5 & .6 & .7 & .8 & .9 & 1\\
\midrule
WS, $n=50$ &0&0&0&3&31&100&167&198&200&200&200\\
WS, $n=200$ &0&0&0&0&0&0&22&97&187&200&200\\
WS, heterogeneous &0&0&0&0&0&12&79&163&198&200&200\\
WS, bimodal &0&0&0&0&0&13&79&163&198&200&200\\
BA, $n=50$ &0&0&0&0&0&1&2&12&44&111&199\\
BA, $n=200$ &0&0&0&0&0&0&0&0&0&4&195\\
BA, heterogeneous &0&0&0&0&0&0&0&0&1&31&200\\
BA, bimodal &0&0&0&0&0&0&0&0&1&31&200\\
\bottomrule
\end{tabular}
\end{table}

Changing network size visibly shifts the frequency curves. At fixed size, changes in coefficients and initial distributions give curves close to their matched baselines in these settings; this does not establish invariance with respect to those parameters. In the all-WA BA case with $n=200$, only 195/200 simulated solutions reach the tolerance by $10^5$ updates. The connected network with bidirectional edges guarantees asymptotic consensus, but the five remaining simulated solutions are still recorded as not reaching the tolerance within the horizon. 

\subsection{Consensus tolerance}
\begin{figure}[!ht]
\centering
\includegraphics[width=\textwidth]{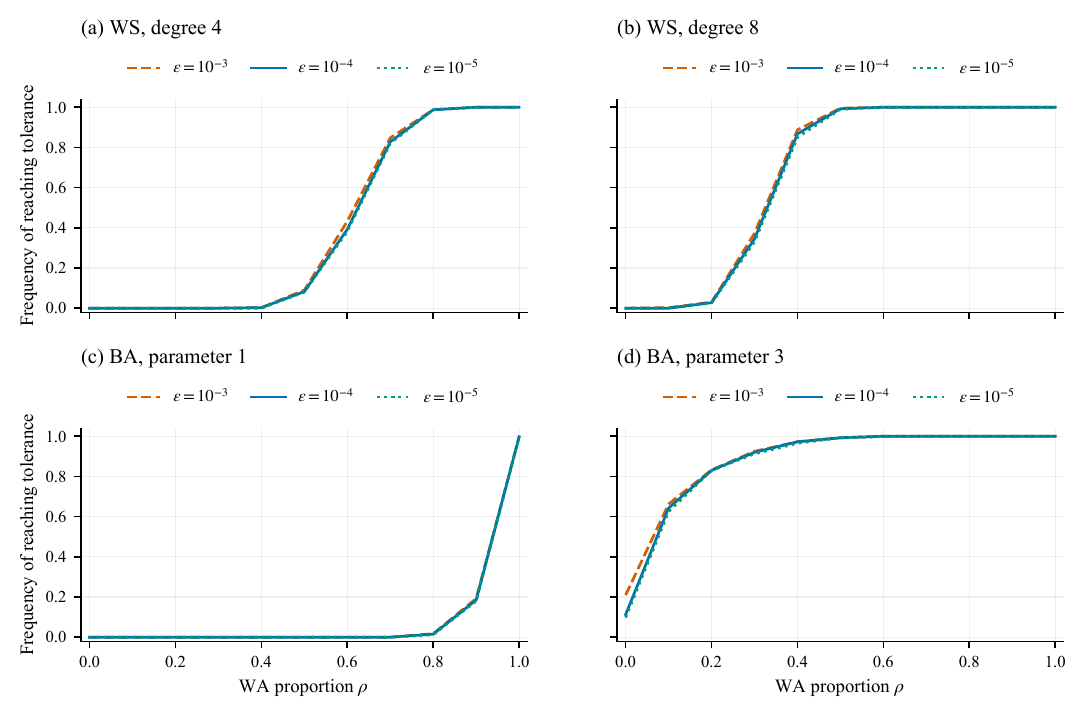}
\caption{Sensitivity to the consensus tolerance in the four network ensembles of Study~3. (a) WS, degree 4. (b) WS, degree 8. (c) BA, attachment 1. (d) BA, attachment 3. All panels use $n=100$, $\gamma_i=0.5$, independent $\mathrm{Unif}([0,1])$ initial opinions, and 400 realizations per ensemble. Orange dashed, blue solid, and green dotted curves show the observed frequencies of reaching $E(k)\leq10^{-3}$, $10^{-4}$, and $10^{-5}$ by $T=10^5$, respectively. Circles mark the eleven sampled WA proportions, and line segments connect adjacent sampled proportions. All curves use the same realizations; confidence intervals are not shown. Network generation and influence weights follow Section~\ref{sec:snetwork}.}\label{fig:tolerance}

\end{figure}
Figure~\ref{fig:tolerance} separately compares tolerances using all 400 main realizations, without generating additional simulated solutions.

In these tested settings, the observed frequency remains nondecreasing with WA proportion, providing further numerical support for the consensus-promoting role of WA individuals. The finite-horizon frequencies nevertheless depend on network size and tolerance. We next examine the time required to reach the tolerance and compare the implementation with an analytical result for the all-WA model.

\section{Hitting-time distributions and an all-WA benchmark}\label{sec:times}
We first examine the time required to reach the tolerance in Study~3. We then use the network of Study~1 with all individuals assigned to WA to check the numerical implementation against an analytical limit and distinguish the roles of inertia coefficients.

\subsection{Hitting times in the WA-proportion study}
Figure~\ref{fig:attainment} plots the fraction of simulated solutions reaching normalized opinion diameter $10^{-4}$ by time $t$. The denominator is always $N=400$, including simulated solutions that do not reach the tolerance, so the curves are not conditioned on reaching the prescribed consensus tolerance. At a fixed WA proportion $\rho$, let $\tau_\varepsilon^{(r)}(\rho)$ denote the hitting time of the simulated solution from realization $r$. The empirical restricted mean hitting time is
\begin{equation}
 \overline{\tau}_{\varepsilon,T}(\rho)
 =\frac1N\sum_{r=1}^N\min\{\tau_\varepsilon^{(r)}(\rho),T\}
 =\sum_{t=0}^{T-1}\left(1-\frac1N\sum_{r=1}^N
       \mathbf1\{\tau_\varepsilon^{(r)}(\rho)\leq t\}\right),
\end{equation}
where $\mathbf1\{\cdot\}$ equals one when the enclosed condition holds and zero otherwise, and $\min\{\infty,T\}=T$. Thus a simulated solution that has not reached the tolerance contributes $T$, and the same quantity is obtained by summing the empirical proportions that have not yet reached it. Here $\varepsilon=10^{-4}$ and $T=10^5$. For a simulation stopped before $T$, the last recorded opinion diameter is not treated as $D(x(T))$.

\begin{figure}[!ht]
\centering
\includegraphics[width=\textwidth]{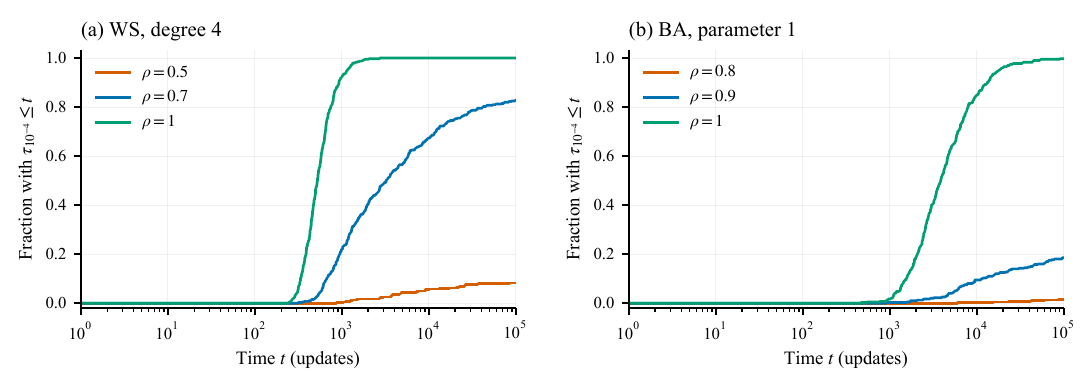}
\caption{Empirical hitting-time distributions for two network ensembles in Study~3. (a) WS, degree 4 and rewiring probability 0.1. (b) BA, attachment 1. Both panels use $n=100$, $\gamma_i=0.5$, independent $\mathrm{Unif}([0,1])$ initial opinions, and 400 paired realizations. Each curve corresponds to the labeled WA proportion and shows the fraction of simulated solutions reaching $E(k)\leq10^{-4}$ by time $t$. All 400 solutions remain in the denominator, including those that have not reached the tolerance. Network generation and influence weights follow Section~\ref{sec:snetwork}.}\label{fig:attainment}

\end{figure}

\subsection{Inertia coefficients in the all-WA model}\label{sec:walimit}
\begin{figure}[tp]
\centering
\includegraphics[width=\textwidth]{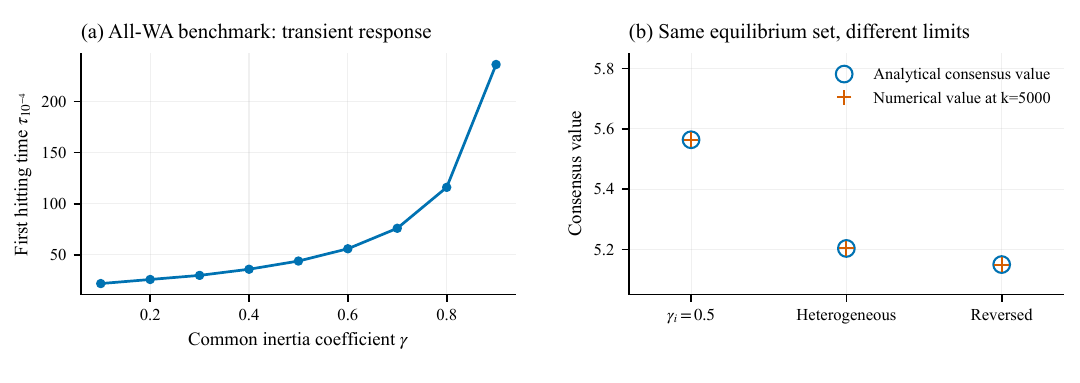}
\caption{All-WA benchmark on the six-node influence network of Study~1, with initial opinions $(8.7594,\allowbreak5.5016,\allowbreak6.2248,\allowbreak5.8704,\allowbreak2.0774,\allowbreak3.0125)^\top$. (a) First hitting time of normalized opinion diameter $10^{-4}$ for common inertia coefficients $0.1,0.2,\ldots,0.9$. (b) Analytical consensus values and numerical values at $k=5000$ for common $\gamma_i=0.5$, heterogeneous $\gamma=(0.3557,\allowbreak0.6214,\allowbreak0.1884,\allowbreak0.4672,\allowbreak0.6493,\allowbreak0.0610)^\top$, and the same vector with its components in reverse order. In panel~(b), each pair of markers compares the analytical consensus value with the numerical value from one deterministic simulated solution; the two values agree to the reported precision.}\label{fig:adjustment}
\end{figure}
The benchmark uses the influence network of Study~1 with all six individuals using WA. This all-WA system is a French--DeGroot model~\cite{french1956,degroot1974} with inertia. For reproducibility, its adjacency matrix is $A=B/2$, where
\[
 B=\begin{pmatrix}
 0&0&1&0&1&0\\
 0&0&2&0&0&0\\
 1&0&0&1&0&0\\
 0&2&0&0&0&0\\
 0&0&1&0&0&1\\
 0&0&0&0&2&0
 \end{pmatrix}.
\]
The initial vector and heterogeneous inertia vector are given in Fig.~\ref{fig:adjustment}. Write $\Gamma=\operatorname{diag}(\gamma_1,\ldots,\gamma_6)$, with $0<\gamma_i<1$, and let $I$ be the $6\times6$ identity matrix. Then the opinion update map~\eqref{eq:supdate} is $F_\gamma(x)=Px$, where
\[
 P=\Gamma+(I-\Gamma)A.
\]
The fixed-point equation $Px=x$ is equivalent to $(I-\Gamma)(A-I)x=0$. Since $I-\Gamma$ is invertible, this is in turn equivalent to $Ax=x$. Consequently, changing the inertia coefficients within the prescribed range does not change the equilibrium set.

The network is strongly connected, so its row-stochastic matrix $A$ is irreducible. Let $\pi$ be its stationary probability vector, defined by $\pi^\top A=\pi^\top$, $\pi>0$, and $\pi^\top\mathbf1=1$, where $\mathbf1$ is the vector of six ones. The stationary probability vector of $P$ and the consensus limit are
\begin{equation}\label{eq:walimit}
 \nu_i=\frac{\pi_i/(1-\gamma_i)}{\sum_j\pi_j/(1-\gamma_j)},\qquad
 x(k)\longrightarrow (\nu^\top x(0))\mathbf1.
\end{equation}
Indeed, $\nu^\top(I-\Gamma)$ is proportional to $\pi^\top$, which gives
\[
 \nu^\top P=\nu^\top+\nu^\top(I-\Gamma)(A-I)=\nu^\top.
\]
Thus $\nu^\top x(k)$ is conserved. To see why consensus is guaranteed, recall that any member of an all-WA type-dependent cohesive set assigns zero weight to its complement. Strong connectivity excludes any proper nonempty set of this kind. There can therefore be no two disjoint nonempty type-dependent cohesive sets, and the consensus criterion stated in Section~\ref{sec:snetwork} applies. Conservation of $\nu^\top x(k)$ then determines the consensus value in~\eqref{eq:walimit}.

If all inertia coefficients are equal, the normalization in~\eqref{eq:walimit} gives $\nu=\pi$. A common coefficient can change transient behavior but leaves the consensus value $\pi^\top x(0)$ unchanged. Unequal coefficients can change that value. This formula applies to the all-WA linear system; it does not give the consensus value for the mixed model.

\begin{samepage}
We test common coefficients $0.1,0.2,\ldots,0.9$, the heterogeneous inertia vector specified above, and that vector with its components in reverse order, computing 5000 updates for each of the 11 cases. Figure~\ref{fig:adjustment} shows hitting times increasing from 22 to 236 updates as the common inertia coefficient increases from 0.1 to 0.9. The common-coefficient consensus value is $5.5637923$; the two heterogeneous values are $5.2030201$ and $5.1495244$. The largest discrepancy between an analytical limit and any final component is $1.45\times10^{-13}$. These results check the linear special case; neither the observed ordering of hitting times nor the consensus-value formula is asserted for arbitrary mixed systems.
\end{samepage}

\section{Code organization and scope of the numerical evidence}\label{sec:reproduction}
\paragraph{Reproduction files}
The model is implemented in \path{dynamics.py} and \path{dynamics_kernel.cpp}. The file \path{structure.py} checks type-dependent cohesive sets, and \path{experiments.py} generates the original simulation batch. Associated tests check both the opinion update and structural routines. The script \path{audit_and_summarize.py} reconstructs the statistics for that batch. Software versions, seeds, protocols, and file hashes are recorded in the numerical archive, which is separate from the manuscript compilation package.

The additional simulation batch is implemented in \path{upgrade/}: \path{protocol.json} records its sampling design, \path{run_upgrade.py} generates the samples, and \path{analyze_plot.py} checks them and constructs the figures. The separate script \path{reversal_diagnostic.py} selects and recomputes the saved case in Section~\ref{sec:reversal}. Its report lists all seven main-study cases that do not follow the aggregate trend, the selection rule, the input hash, and the precision checks. Neither the sampling protocol nor any original sample is modified by this diagnostic. Figures use vector lines and embedded fonts at their intended physical dimensions. Lower plotting limits affect only the display, not the numerical data.

\paragraph{Scope of the evidence}
Studies~3 and~4 and the sensitivity analyses use synthetic networks with bidirectional edge sets, independently sampled weights in each direction, and fixed WA--WM configurations. The six-node influence network of Study~1 and the five-node community construction of Study~2 are specified separately. Together, the studies examine WA locations and proportions, influence weights, and the relation between reaching the prescribed consensus tolerance within the simulation horizon and the structural guarantee of consensus. The empirical frequency trends do not establish monotonicity for every realization, a universal critical WA proportion, or a uniform convergence-time bound. These simulations use no observations of human opinions.